\documentclass[11pt]{article}
\usepackage[margin=1in]{geometry}

\usepackage{amsmath, amsthm, amssymb, microtype, tikz, framed}
\usetikzlibrary{shapes}
\usepackage[colorlinks, citecolor=blue, linkcolor=blue, bookmarks=true]{hyperref}
\usepackage[nameinlink, capitalize]{cleveref}

\newtheorem{prop}{Proposition}
\newtheorem{thm}{Theorem}
\newtheorem{lem}{Lemma}

\newtheorem{cor}{Corollary}
\theoremstyle{definition}
\newtheorem{defn}{Definition}

\newcommand{\R}{\mathbb{R}}

\newcommand{\N}{\mathbb{N}}
\newcommand{\Z}{\mathbb{Z}}

\renewcommand{\epsilon}{\varepsilon}

\renewcommand{\tilde}{\widetilde}
\renewcommand{\bar}{\overline}
\renewcommand{\hat}{\widehat}

\DeclareMathOperator{\poly}{poly}
\DeclareMathOperator*{\E}{E}
\DeclareMathOperator*{\median}{median}

\mathchardef\mhyphen="2D

\author{William M. Hoza\thanks{Supported by the NSF GRFP under Grant No. DGE-1610403 and by a Harrington Fellowship from UT Austin.}\\Department of Computer Science,\\University of Texas at Austin\\\texttt{whoza@utexas.edu} \and Adam R. Klivans\\Department of Computer Science,\\University of Texas at Austin\\\texttt{klivans@cs.utexas.edu}}
\title{Preserving Randomness for Adaptive Algorithms}

\begin{document}
	\maketitle

	\begin{abstract}
		Suppose $\mathsf{Est}$ is a randomized estimation algorithm that uses $n$ random bits and outputs values in $\R^d$. We show how to execute $\mathsf{Est}$ on $k$ adaptively chosen inputs using only $n + O(k \log(d + 1))$ random bits instead of the trivial $nk$ (at the cost of mild increases in the error and failure probability). Our algorithm combines a variant of the INW pseudorandom generator \cite{inw94} with a new scheme for shifting and rounding the outputs of $\mathsf{Est}$. We prove that modifying the outputs of $\mathsf{Est}$ is necessary in this setting, and furthermore, our algorithm's randomness complexity is near-optimal in the case $d \leq O(1)$. As an application, we give a randomness-efficient version of the Goldreich-Levin algorithm; our algorithm finds all Fourier coefficients with absolute value at least $\theta$ of a function $F: \{0, 1\}^n \to \{-1, 1\}$ using $O(n \log n) \cdot \poly(1/\theta)$ queries to $F$ and $O(n)$ random bits (independent of $\theta$), improving previous work by Bshouty et al.\ \cite{bjt04}.
	\end{abstract}
	
	\section{Introduction}
	
	Let $\mathsf{Est}$ be a randomized algorithm that estimates some quantity $\mu(C) \in \R^d$ when given input $C$. The canonical example is the case when $C$ is a Boolean circuit, $d = 1$, $\mu(C) \stackrel{\text{def}}{=} \Pr_x[C(x) = 1]$, and $\mathsf{Est}$ estimates $\mu(C)$ by evaluating $C$ at several randomly chosen points. Suppose that $\mathsf{Est}$ uses $n$ random bits, and $\Pr[\|\mathsf{Est}(C) - \mu(C)\|_{\infty} > \epsilon] \leq \delta$.
	
	Furthermore, suppose we want to use $\mathsf{Est}$ as a \emph{subroutine}, executing it on inputs $C_1, C_2, \dots, C_k$, where each $C_i$ is chosen adaptively based on the previous outputs of $\mathsf{Est}$. The na{\"i}ve implementation uses $nk$ random bits and fails with probability at most $k\delta$.
	
	In this work, we show how to generically improve the randomness complexity of any algorithm with this structure, without increasing the number of executions of $\mathsf{Est}$, at the expense of mild increases in the error and failure probability. Our algorithm efficiently finds $Y_1, \dots, Y_k \in \R^d$ with $\|Y_i - \mu(C_i)\|_{\infty} \leq O(\epsilon d)$ for every $i$, our algorithm has failure probability $k\delta + \gamma$ for any $\gamma > 0$, and our algorithm uses a total of $n + O(k \log(d + 1) + (\log k) \log(1/\gamma))$ random bits.
	
	\subsection{The randomness steward model} \label{sec:stewardship-problem}
	
	We model the situation described above by imagining two interacting agents: the \emph{owner} (who plays the role of the outer algorithm) chooses the inputs $C_1, \dots, C_k$, while the \emph{steward} (who replaces the direct execution of $\mathsf{Est}$) provides the output vectors $Y_1, \dots, Y_k \in \R^d$. The reader might find it useful to think of the steward as a trusted ``cloud computing'' service. To justify the names, imagine that the owner gives the steward ``stewardship'' over her random bits. The steward's job is to ``spend'' as little randomness as possible without sacrificing too much accuracy.
	
	To describe the model more rigorously, say that a function $f: \{0, 1\}^n \to \R^d$ is \emph{$(\epsilon, \delta)$-concentrated at $\mu \in \R^d$} if $\Pr_{X \in \{0, 1\}^n}[\|f(X) - \mu\|_{\infty} > \epsilon] \leq \delta$. In each round $i$, the chosen input $C_i$ defines a concentrated function $f_i(X) \stackrel{\text{def}}{=} \mathsf{Est}(C_i, X)$, so it is equivalent to imagine that the owner picks an arbitrary concentrated function. In the following definition, $\epsilon'$ is the error of the steward, and $\delta'$ is its failure probability.
	
	\begin{samepage}
	\begin{defn} \label{defn:steward}
		An \emph{$(\epsilon', \delta')$-steward for $k$ adaptively chosen $(\epsilon, \delta)$-concentrated functions $f_1, \dots, f_k: \{0, 1\}^n \to \R^d$} is a randomized algorithm $\mathsf{S}$ that interacts with an \emph{owner} $\mathsf{O}$ according to the following protocol.
		\begin{enumerate}
			\item For $i = 1$ to $k$:
			\begin{enumerate}
				\item $\mathsf{O}$ chooses $f_i: \{0, 1\}^n \to \R^d$ that is $(\epsilon, \delta)$-concentrated at some point $\mu_i \in \R^d$ and gives it to $\mathsf{S}$.
				\item $\mathsf{S}$ chooses $Y_i \in \R^d$ and gives it to $\mathsf{O}$.
			\end{enumerate}
		\end{enumerate}
		Write $\mathsf{O} \leftrightarrow \mathsf{S}$ (``the interaction of $\mathsf{O}$ with $\mathsf{S}$'') to denote the above interaction. The requirement on $\mathsf{S}$ is that for all $\mathsf{O}$,
		\[
			\Pr[\max_i \|Y_i - \mu_i\|_{\infty} > \epsilon' \text{ in } \mathsf{O} \leftrightarrow \mathsf{S}] \leq \delta'.
		\]
		The probability is taken over the internal randomness of $\mathsf{S}$ and $\mathsf{O}$.
	\end{defn}
	\end{samepage}

	From an information-theoretic perspective, stewards as defined above are not particularly interesting, because $\mathsf{S}$ could exhaustively examine all outputs of $f_i$ to deterministically compute a point $Y_i$ where $f_i$ is concentrated. But we would like to avoid executing $\mathsf{Est}$ more than $k$ times in total, so we will restrict attention to \emph{one-query stewards}:
	\begin{defn}
		A \emph{one-query steward} is a steward that only accesses each $f_i$ by querying it at a single point $X_i \in \{0, 1\}^n$. (The point $X_i$ is not seen by the owner.)
	\end{defn}
	
	\subsection{Our results}
	
	\subsubsection{Main result: A one-query steward with good parameters}
	Our main result is the explicit construction of a one-query steward that simultaneously achieves low error, low failure probability, and low randomness complexity:
	\begin{thm} \label{thm:main}
		For any $n, k, d \in \N$ and any $\epsilon, \delta, \gamma > 0$, there exists a one-query $(O(\epsilon d), k\delta + \gamma)$-steward for $k$ adaptively chosen $(\epsilon, \delta)$-concentrated functions $f_1, \dots, f_k: \{0, 1\}^n \to \R^d$ with randomness complexity
		\[
			n + O(k\log(d + 1) + (\log k)\log(1/\gamma)).
		\]
		The total running time of the steward is $\poly(n, k, d, \log(1/\epsilon), \log(1/\gamma))$.
	\end{thm}
	We also give several variant stewards that achieve tradeoffs in parameters. (See \cref{fig:stewards}.)
	
	\begin{figure}
		\begin{center}
			\def\arraystretch{1.1}%
			\crefname{prop}{Prop.}{Props.}
			\begin{tabular}{l|l|l|l}
				$\epsilon'$ & $\delta'$ & Randomness complexity & Reference \\
				\hline
				$\epsilon$ & $k\delta$ & $nk$ & Na{\"i}ve \\
				$O(\epsilon d)$ & $k\delta + \gamma$ & $n + O(k \log(d + 1) + (\log k) \log(1/\gamma))$ & \cref{thm:main} (main) \\
				$O(\epsilon)$ & $k\delta + \gamma$ & $n + O(kd + (\log k) \log(1/\gamma))$ & \cref{thm:prg-steward}, $d_0 = 1$ \\
				$O(\epsilon d)$ & $k\delta + \gamma$ & $n + k\log(d + 2) + 2\log(1/\gamma) + O(1)$ & \cref{thm:computationally-inefficient}* \\
				$O(\epsilon d)$ & $2^{O(k \log(d + 1))} \cdot \delta$ & $n$ & \cref{thm:union-bound}, $d_0 = d$ \\
				$O(\epsilon)$ & $2^{O(kd)} \cdot \delta$ & $n$ & \cref{thm:union-bound}, $d_0 = 1$ \\
				$O(\epsilon k d / \gamma)$ & $k\delta + \gamma$ & $n + O(k \log k + k \log d + k \log(1/\gamma))$ & \cref{prop:sz} (based on \cite{sz99}) \\
				$O(\epsilon)$ & $k\delta + k \cdot 2^{-n^{\Omega(1)}}$ & $O(n^6 + kd)$ & \cref{prop:iz} (based on \cite{iz89}) \\
				Any & Any $\leq 0.2$ & $n + \Omega(k) - \log(\delta'/\delta)$ & \cref{thm:lower-bound} (lower bound)
			\end{tabular} \\
			\crefname{prop}{Proposition}{Propositions}
		\end{center}
		*Computationally inefficient.
		\caption{Upper and lower bounds for one-query stewards. Recall that $\epsilon, \delta$ are the concentration parameters of $f_1, \dots, f_k$ (i.e. the error and failure probability of the estimation algorithm $\mathsf{Est}$); $\epsilon', \delta'$ are the error and failure probability of the steward $\mathsf{S}$; $n$ is the number of input bits to each $f_i$ (i.e. the number of random coins used by $\mathsf{Est}$); $k$ is the number of rounds of adaptivity; $d$ is the dimension of the output of each $f_i$ (i.e. the dimension of the output of $\mathsf{Est}$). Everywhere it appears, $\gamma$ denotes an arbitrary positive number.} \label{fig:stewards}
	\end{figure}
	
	\subsubsection{Application: Acceptance probabilities of Boolean circuits}
	Our first concrete application of \cref{thm:main} is a time- and randomness-efficient algorithm for estimating the acceptance probabilities of many adaptively chosen Boolean circuits.
	\begin{cor} \label{cor:circuit}
		There exists a randomized algorithm with the following properties. Initially, the algorithm is given parameters $n, k \in \N$ and $\epsilon, \delta > 0$. Then, in round $i$ ($1 \leq i \leq k$), the algorithm is given a Boolean circuit $C_i$ on $n$ input bits and outputs a number $Y_i \in [0, 1]$. Here, $C_i$ may be chosen adversarially based on $Y_1, \dots, Y_{i - 1}$. With probability $1 - \delta$, every $Y_i$ is $\mu(C_i) \pm \epsilon$, where $\mu(C_i) \stackrel{\mathrm{def}}{=} \Pr_x[C_i(x) = 1]$. The total running time of the algorithm is
		\[
			O\left(\frac{\log k + \log(1/\delta)}{\epsilon^2} \cdot \sum_{i = 1}^k \mathrm{size}(C_i)\right) + \poly(n, k, 1/\epsilon, \log(1/\delta)),
		\]
		and the total number of random bits used by the algorithm is $n + O(k + (\log k) \cdot \log(1/\delta))$.
	\end{cor}
	\Cref{cor:circuit} should be compared to the case when $C_1, \dots, C_k$ are chosen nonadaptively, for which the randomness complexity can be improved to $n + O(\log k + \log(1/\delta))$ by applying the Goldreich-Wigderson randomness-efficient sampler for Boolean functions \cite{gw97} and reusing randomness. The proof of \cref{cor:circuit} works by combining the GW sampler with our steward.
	
	\subsubsection{Application: Simulating an oracle for $\mathbf{promise\mhyphen BPP}$ or $\mathbf{APP}$} \label{sec:intro-oracle}
	Recall that $\mathbf{promise\mhyphen BPP}$ is the class of promise problems that can be decided in probabilistic polynomial time with bounded failure probability. When an algorithm is given oracle access to a promise problem, it is allowed to make queries that violate the promise, and several models have been considered for dealing with such queries. Following Moser \cite{mos01}, we will stipulate that the oracle may respond in any arbitrary way to such queries. (See, e.g., \cite{bf99} for two other models.) From these definitions, it is easy to show, for example, that $\mathbf{BPP}^{\mathbf{promise\mhyphen BPP}} = \mathbf{BPP}$. Using our steward, we give a time- and randomness-efficient simulation of any algorithm with an oracle for $\mathbf{promise\mhyphen BPP}$. (As we will discuss in \cref{sec:pseudodeterminism}, the corresponding result for $\mathbf{BPP}$-oracle algorithms is trivial.) 
	We also give a similar result for algorithms with an oracle for $\mathbf{APP}$, the class introduced by Kabanets et al.\ \cite{krc00} of functions $\varphi: \{0, 1\}^n \to [0, 1]$ that can be approximated to within $\pm \epsilon$ in probabilistic $\poly(n, 1/\epsilon)$ time with bounded failure probability.
	
	\subsubsection{Application: The Goldreich-Levin algorithm}
	As a final application, we give a randomness-efficient version of the Goldreich-Levin algorithm \cite{gl89} (otherwise known as the Kushilevitz-Mansour algorithm \cite{km93}) for finding noticeably large Fourier coefficients. Given oracle access to $F: \{0, 1\}^n \to \{-1, 1\}$, for any $\theta > 0$, we show how to efficiently find a list containing all $U$ with $|\hat{F}(U)| \geq \theta$. (Alternatively, thinking of $F$ as an exponentially long bitstring $F \in \{-1, 1\}^{2^n}$, our algorithm finds all Hadamard codewords that agree with $F$ in a $\left(\frac{1}{2} + \theta\right)$-fraction of positions.) Our algorithm makes $O(n \log(n/\delta)) \cdot \poly(1/\theta)$ queries to $F$, uses $O(n + (\log n) \log(1/\delta))$ random bits, and has failure probability $\delta$. The number of random bits \emph{does not depend on $\theta$}. To achieve such a low randomness complexity, we first improve the randomness efficiency of each estimate in the Goldreich-Levin algorithm using the GW sampler. Then, we reduce the number of rounds of adaptivity by a factor of $\log(1/\theta)$ by making many estimates within each round. Interestingly, we apply our steward with $d = \poly(1/\theta)$, unlike our other applications where we choose $d = 1$. (Recall that $d$ is the number of real values estimated in each round.)
	
	
	\subsubsection{Randomness complexity lower bound}
	We prove a randomness complexity lower bound of $n + \Omega(k) - \log(\delta'/\delta)$ for any one-query steward. In the case $d \leq O(1)$, this comes close to matching our upper bounds. For example, to achieve $\delta' \leq O(k \delta)$, this lower bound says that $n + \Omega(k)$ random bits are needed; our main steward (\cref{thm:main}) achieves $\epsilon' \leq O(\epsilon), \delta' \leq O(k \delta)$ using $n + O(k + (\log k) \log(1/\delta))$ random bits. At the other extreme, if we want a one-query steward that uses only $n$ random bits, this lower bound says that the failure probability will be $\delta' \geq \exp(\Omega(k)) \cdot \delta$; one of our variant stewards (\cref{thm:union-bound}) uses $n$ random bits to achieve $\epsilon' \leq O(\epsilon)$ and $\delta' \leq \exp(O(k)) \cdot \delta$.
	
	\subsection{Techniques} \label{sec:techniques}
	
	\subsubsection{Block decision trees}
	
	A key component in the proof of our main result (\cref{thm:main}) is a pseudorandom generator (PRG) for a new model that we call the \emph{block decision tree} model. Informally, a block decision tree is a decision tree that reads its input from left to right, $n$ bits at a time:
		
	\begin{defn}
		For a finite alphabet $\Sigma$, a \emph{$(k, n, \Sigma)$ block decision tree} is a rooted tree $T = (V, E)$ of height $k$ in which every node $v$ at depth $< k$ has exactly $|\Sigma|$ children (labeled with the symbols in $\Sigma$) and has an associated function $v: \{0, 1\}^n \to \Sigma$. We identify $T$ with a function $T: (\{0, 1\}^n)^{\leq k} \to V$ defined recursively: $T(\text{the empty string}) = \text{the root node}$, and if $T(X_1, \dots, X_{i - 1}) = v$, then $T(X_1, \dots, X_i)$ is the child of $v$ labeled $v(X_i)$.
	\end{defn}
		
	The standard nonconstructive argument (\cref{apx:prg}) shows that there exists a $\gamma$-PRG for block decision trees with seed length $n + k \log |\Sigma| + 2 \log(1/\gamma) + O(1)$. (See \cref{sec:prg-formalities} for the definition of a PRG in this setting.) In \cref{sec:prg}, we explicitly construct a $\gamma$-PRG for block decision trees with seed length $n + O(k \log |\Sigma| + (\log k) \log(1/\gamma))$. The generator is constructed by modifying the INW generator for space-bounded computation \cite{inw94}.
	
	\subsubsection{Shifting and rounding} \label{sec:techniques-shifting-and-rounding}

	\paragraph{PRGs alone are not enough} Interestingly, to design good stewards, the standard technique of replacing truly random bits with pseudorandom bits is not sufficient. That is, define a \emph{pseudorandom generation steward} to be a steward that simply queries each $f_i$ at a pseudorandomly chosen point $X_i$ and returns $Y_i = f_i(X_i)$. In \cref{sec:prg-steward}, we show that any pseudorandom generation steward must use at least $\Omega(nk)$ random bits, assuming $\delta' \leq 1/2$ and $\delta \geq 2^{-n/2 + 1}$.
	
	\paragraph{Deliberately introducing error} To circumvent this $\Omega(nk)$ lower bound, our steward is forced to \emph{modify the outputs} of $f_1, \dots, f_k$. In each round, our main steward queries $f_i$ at a pseudorandom point $X_i$, chooses a nearby value $Y_i \approx f_i(X_i)$, and returns this modified value $Y_i$ to the owner $\mathsf{O}$. The motivation for this idea is that by deliberately introducing a small amount of error, we reduce the amount of information about $X_i$ that is leaked by $Y_i$. This way, we can recycle some of the randomness of $X_i$ for future rounds.
	
	To be more specific, for a steward $\mathsf{S}$, let $\mathsf{S}(X)$ denote $\mathsf{S}$ using randomness $X$. Our main steward is of the form $\mathsf{S}(X) \stackrel{\text{def}}{=} \mathsf{S}_0(\mathsf{Gen}(X))$. Here, $\mathsf{Gen}$ is our PRG for block decision trees, and $\mathsf{S}_0$ is a randomness-inefficient one-query steward. In each round, $\mathsf{S}_0$ queries $f_i$ at a fresh random point $X_i \in \{0, 1\}^n$, but $\mathsf{S}_0$ computes the return value $Y_i$ by carefully \emph{shifting and rounding} each coordinate of $f_i(X_i)$. This deterministic shifting and rounding procedure and its analysis are our main technical contributions. 
	
	\paragraph{Relation to block decision trees} To explain why this works, observe that when any steward and owner interact, it is natural to model the owner's behavior by a decision tree that branches at each node based on the value $Y_i$ provided by the steward. The \emph{branching factor} of this decision tree is a simple measure of the amount of information leaked. If $\mathsf{S}_0$ simply returned $f_i(X_i)$ without any shifting or rounding, the branching factor for $\mathsf{O} \leftrightarrow \mathsf{S}_0$ would be $2^n$. Three ideas dramatically reduce this branching factor.
	\begin{itemize}
		\item The first idea is to round. Suppose that $\mathsf{S}_0$ rounded each coordinate of $f_i(X_i)$ to the nearest multiple of $2\epsilon$ (with no shifting). Then the branching factor would be reduced to $2^d + \delta 2^n$. The $\delta 2^n$ term corresponds to the outputs of $f_i$ that are far from its concentration point $\mu_i$. The $2^d$ term corresponds to outputs $f_i(X_i)$ that are close to $\mu_i$; if each coordinate of $\mu_i$ is approximately equidistant from two multiples of $2\epsilon$, then the corresponding coordinate of $f_i(X_i)$ could be rounded to either of those two values.
		\item The second idea is to \emph{shift} each coordinate of $f_i(X_i)$ before rounding. In particular, $\mathsf{S}_0$ finds a single value $\Delta_i$ such that after adding $\Delta_i \cdot 2\epsilon$ to each coordinate of $f_i(X_i)$, every coordinate is $\epsilon$-far from every rounding boundary. Then, $\mathsf{S}_0$ rounds the shifted coordinates to obtain $Y_i$. This procedure reduces the branching factor down to $d + 1 + \delta 2^n$. To understand why, think of $\Delta_i$ as a \emph{compressed} representation of $Y_i$. Assuming $f_i(X_i)$ is close to $\mu_i$, given unlimited computation time, $\mathsf{O}$ could recover $Y_i$ from $\Delta_i$ by computing the \emph{true} vector $\mu_i$, shifting \emph{it} according to $\Delta_i$, and rounding. Hence, each node of the tree just needs to have one child for each possible $\Delta_i$ value (along with the $\delta 2^n$ children for the case that $f_i(X_i)$ is far from $\mu_i$).
		\item The third idea is to \emph{relax} the requirement that the tree perfectly computes $\mathsf{O} \leftrightarrow \mathsf{S}_0$. In particular, for every owner $\mathsf{O}$, we construct a block decision tree $T_{\mathsf{O}}$ that merely \emph{certifies correctness} of $\mathsf{O} \leftrightarrow \mathsf{S}_0$. That is, for any $X_1, \dots, X_k$, if the node $T_{\mathsf{O}}(X_1, \dots, X_k)$ indicates ``success'', then the error $\max_i \|Y_i - \mu_i\|_{\infty}$ in $\mathsf{O} \leftrightarrow \mathsf{S}_0(X_1, \dots, X_k)$ is small. On the other hand, if $T_{\mathsf{O}}(X_1, \dots, X_k)$ does not indicate success, then ``all bets are off'': the error $\max_i \|Y_i - \mu_i\|_{\infty}$ in $\mathsf{O} \leftrightarrow \mathsf{S}_0(X_1, \dots, X_k)$ may be small or large. Our certification tree has the additional property that
		\[
		\Pr_{X_1, \dots, X_k}[T_{\mathsf{O}}(X_1, \dots, X_k) \text{ indicates success}] \geq 1 - k\delta.
		\]
		This relaxation allows us to reduce the branching factor down to just $d + 2$, because for each node, the $\delta 2^n$ children corresponding to outputs of $f_i$ that are far from $\mu_i$ can all be merged into a single ``failure'' node.
	\end{itemize}
	
	Putting everything together, to save random bits, we don't need to try to fool $\mathsf{O} \leftrightarrow \mathsf{S}_0$. Instead, it suffices for $\mathsf{Gen}$ to fool the certification tree $T_{\mathsf{O}}$. The small branching factor of $T_{\mathsf{O}}$ allows $\mathsf{Gen}$ to have a correspondingly small seed length.
	
	\subsubsection{Lower bound}
	
	
	Our lower bound for general one-query stewards follows a similar intuition as our upper bounds: we show that in each round, by carefully choosing $f_i$, the owner can learn $\Omega(1)$ bits of information about the steward's randomness. To conclude, we argue that if the steward has fewer than $n$ bits of randomness remaining from the owner's perspective, then the owner can choose a function that causes the steward's failure probability to be large.
	
	\subsection{Why can't we just reuse the random bits?} \label{sec:pseudodeterminism}
	
	Notwithstanding our lower bounds, the reader might be tempted to think that randomness stewards are trivial: why not just pick $X \in \{0, 1\}^n$ uniformly at random \emph{once} and reuse it in every round? For the purpose of discussion, let us generalize, and suppose we are trying to execute an $n$-coin algorithm $\mathsf{A}$ (not necessarily an estimation algorithm) on $k$ inputs $C_1, \dots, C_k$. If $C_1, \dots, C_k$ are chosen \emph{non-adaptively} (i.e. all in advance), then we really can use the same $X$ for each execution. By the union bound, the probability that $\mathsf{A}(C_i, X)$ fails for any $i$ is at most $k \delta$.
	
	That argument breaks down in the adaptive case, because $C_2$ is chosen based on $\mathsf{A}(C_1, X)$, and hence $C_2$ may be \emph{stochastically dependent} on $X$, so $\mathsf{A}(C_2, X)$ is not guaranteed to have a low failure probability. For example, if $X$ is encoded in the output $\mathsf{A}(C_1, X)$, then an adversarially chosen $C_2$ could \emph{guarantee} that $\mathsf{A}(C_2, X)$ fails. 
	
	Even if $C_1, \dots, C_k$ are chosen adaptively, randomness can be safely reused in an important special case: Suppose $\mathsf{A}$ is a $\mathbf{BPP}$ algorithm. Then we can let $\hat{C}_1, \hat{C}_2, \dots, \hat{C}_k$ be the inputs that would be chosen if $\mathsf{A}$ never failed. Then each $\hat{C}_i$ really is independent of $X$, so by the union bound, with probability $1 - k\delta$, $\mathsf{A}(\hat{C}_i, X)$ does not fail for any $i$. But if $\mathsf{A}(\hat{C}_i, X)$ does not fail for any $i$, then by induction, $C_i = \hat{C}_i$ for every $i$. So the overall failure probability is once again at most $k \delta$.
	
	More generally, randomness can be safely reused if $\mathsf{A}$ is \emph{pseudodeterministic}, i.e. for each input, there is a unique correct output that $\mathsf{A}$ gives with probability $1 - \delta$.\footnote{These two conditions (inputs are chosen nonadaptively, $\mathsf{A}$ is pseudodeterministic) are both special cases of the following condition under which the randomness $X$ may be safely reused: for every $1 \leq i \leq k$, $C_i$ is a pseudodeterministic function of $(C_0, C_1, \dots, C_{i - 1})$, where $C_0$ is a random variable that is independent of $X$.} (Pseudodeterministic algorithms were introduced by Gat and Goldwasser \cite{gg11}.) A $\mathbf{BPP}$ algorithm is trivially pseudodeterministic.
	
	Observe, however, that a $\mathbf{promise\mhyphen BPP}$ algorithm is only guaranteed to be pseudodeterministic on inputs that satisfy the promise. This is why the result we mentioned in \cref{sec:intro-oracle} is in terms of an oracle for $\mathbf{promise\mhyphen BPP}$. Similarly, estimation algorithms (including $\mathbf{APP}$ algorithms) are typically not pseudodeterministic.
	
	In the standard Goldreich-Levin algorithm, randomness is used to estimate $\sum_{U \in \mathcal{U}} \hat{F}(U)^2$ for certain collections of subsets $\mathcal{U}$. The algorithm's behavior depends on how the estimate compares to $\theta^2/2$. This process is not pseudodeterministic, because if the true value $\sum_{U \in \mathcal{U}} \hat{F}(U)^2$ is very close to $\theta^2/2$, the estimate falls on each side of $\theta^2/2$ with noticeable probability.
	
	\subsection{Related work}
	
	\subsubsection{Adaptive data analysis}
	The notion of a randomness steward is inspired by the closely related \emph{adaptive data analysis} problem, introduced by Dwork et al.\ \cite{dfh+15}. In the simplest version of this problem, there is an unknown distribution $\mathcal{D}$ over $\{0, 1\}^n$ and a \emph{data analyst} who wishes to estimate the mean values (with respect to $\mathcal{D}$) of $k$ adaptively chosen functions $f_1, \dots, f_k: \{0, 1\}^n \to [0, 1]$ using as few samples from $\mathcal{D}$ as possible. In this setting, these samples are held by a \emph{mechanism} and are not directly accessible by the data analyst. In round $i$, the data analyst gives $f_i$ to the mechanism, and the mechanism responds with an estimate of $\E_{x \sim \mathcal{D}}[f_i(x)]$. The mechanism constructs the estimate so as to leak as little information as possible about the sample, so that the same sample points can be safely reused for future estimates.
	
	The data analyst and mechanism in the adaptive data analysis setting are analogous to the owner $\mathsf{O}$ and steward $\mathsf{S}$ in our setting, respectively. In each case, the idea is that the mechanism or steward can intentionally introduce a small amount of error into each estimate to hide information and thereby facilitate future estimates. Note, however, that in the adaptive data analysis problem, there is just one unknown distribution $\mathcal{D}$ and we are concerned with sample complexity, whereas in the randomness stewardship problem, we can think of each concentrated function $f_i$ as defining a new distribution over $\R^d$ and we are concerned with randomness complexity.
	
	\subsubsection{The Saks-Zhou algorithm}
	Another highly relevant construction is the algorithm of Saks and Zhou \cite{sz99} for simulating randomized log-space algorithms in deterministic space $O(\log^{3/2} n)$. The key component in this algorithm can be reinterpreted as a one-query randomness steward. Saks and Zhou also constructed a randomized algorithm $\mathsf{Est}$ that approximates a large power of a given substochastic matrix using Nisan's pseudorandom generator \cite{nis92}. (In fact, Nisan's generator can be replaced with any pseudorandom generator for small space \cite{arm98, hu17}.) The Saks-Zhou algorithm works by applying $\mathsf{Est}$ repeatedly to approximate a much larger power of a given substochastic matrix. The ``Saks-Zhou steward'' reduces the randomness complexity of this process.
	
	The Saks-Zhou steward works by \emph{randomly perturbing and rounding} the output of each $f_i$, and then reusing the same random query point $X$ in each round. The perturbation and rounding are somewhat similar to our construction, but note that we shift the outputs of each $f_i$ deterministically, whereas the Saks-Zhou steward uses random perturbations. The analysis of the Saks-Zhou steward is substantially different than the analysis of our steward. Instead, the analysis of the Saks-Zhou steward is similar to the proof that randomness can be safely reused for a pseudodeterministic subroutine; one can show that random perturbation and rounding effectively breaks the dependence between $X$ and $Y_i$. (See \cref{apx:sz} for the description and analysis of the Saks-Zhou steward.)
	
	Our steward achieves better parameters than the Saks-Zhou steward (see \cref{fig:stewards}). In particular, to achieve failure probability $k\delta + \gamma$, the error $\epsilon'$ of the Saks-Zhou steward is $O(\epsilon k d / \gamma)$ -- the error grows linearly with $k$, the number of rounds of adaptivity, as well as with $1/\gamma$. This implies, for example, that if we tried to use the Saks-Zhou steward to estimate the acceptance probabilities of $k$ adaptively chosen Boolean circuits to within $\pm \epsilon$ with failure probability $\delta$ in a randomness-efficient way, we would need to evaluate each circuit on $\Theta(k^2 \delta^{-2} \epsilon^{-2} \log(k/\delta))$ inputs. In contrast, because of our steward's low error, the algorithm of \cref{cor:circuit} evaluates each circuit on just $O(\epsilon^{-2} \log(k/\delta))$ inputs -- an exponential improvement in both $k$ and $1/\delta$. Furthermore, our steward has better randomness complexity than the Saks-Zhou steward.
	
	\subsubsection{Pseudorandom generators for adaptive algorithms} \label{sec:iz}
	Impagliazzo and Zuckerman \cite{iz89, imp92} were the first to consider the problem of saving random bits when executing a randomized algorithm $\mathsf{A}$ on many adaptively chosen inputs. Instead of assuming that $\mathsf{A}$ is an estimation algorithm, Impagliazzo and Zuckerman's result assumes a known bound on the Shannon entropy of the output distribution of $\mathsf{A}$ (e.g., the number of bits output by $\mathsf{A}$). They constructed a pseudorandom generator for this setting; for $k \gg n^6$, the seed length is approximately the sum of the entropy bounds for all the executions of $\mathsf{A}$.
	
	In contrast, we make no assumptions about the entropy of $\mathsf{Est}(C)$. Since $\mathsf{Est}(C)$ is a vector of arbitrary-precision real numbers, the entropy could be as large as $n$, the number of random bits used by $\mathsf{Est}$. And indeed, our lower bound in \cref{sec:prg-steward} implies that the approach of Impagliazzo and Zuckerman fails in our setting.
	
	One might protest that the entropy of $\mathsf{Est}(C)$ can be reduced by simple rounding. In \cref{apx:iz}, we construct and analyze a steward that straightforwardly rounds each output and then uses the Impagliazzo-Zuckerman generator in a black-box way. Our main steward achieves much better randomness complexity and failure probability than this ``Impagliazzo-Zuckerman steward'' (see \cref{fig:stewards}). Our main steward admittedly has larger error than the Impagliazzo-Zuckerman steward ($O(\epsilon d)$ vs. $O(\epsilon)$), but one of our variant stewards beats or matches the Impagliazzo-Zuckerman steward in every parameter. (See \cref{apx:iz-comparison} for details.) The improvements come from our more powerful PRG and the fact that we shift before rounding.
	
	\subsubsection{Decision trees and branching programs}
	In the most common decision tree model, the branching factor $|\Sigma|$ is just $2$, and each node reads an arbitrary bit of the input. In the more general \emph{parity decision tree} model, each node computes the parity of some subset of the input bits. Kushilevitz and Mansour showed \cite{km93} that the Fourier $\ell_1$ norm of any Boolean function computed by a parity decision tree is at most $2^k$, the number of leaves in the tree. It follows immediately that a $\gamma$-biased generator is a $(2^k \gamma)$-PRG for parity decision trees. Using, e.g., the small-bias generator of Naor and Naor \cite{nn93}, this gives an efficient PRG for parity decision trees with asymptotically optimal seed length.
	
	Decision trees in which each node computes a more complicated function have also been studied previously. Bellare \cite{bel92} introduced the \emph{universal decision tree} model, in which each node computes an arbitrary Boolean function of the input bits. He gave a bound on the $\ell_1$ norm of any Boolean function computed by a universal decision tree in terms of the $\ell_1$ norms of the functions at each node. Unfortunately, for block decision trees, his bound is so large that it does not immediately imply any nontrivial pseudorandom generators for block decision trees.
	
	A block decision tree can be thought of as a kind of space-bounded computation. Indeed, a block decision tree is a specific kind of \emph{ordered branching program} of width $|\Sigma|^k$ and length $k$ that reads $n$ bits at a time. Hence, we could directly apply a pseudorandom generator for ordered branching programs, such as the INW generator \cite{inw94}. For these parameters, the INW generator has seed length of $n + O(k \log k \log |\Sigma| + \log k \log(1/\gamma))$.
	This seed length can be slightly improved by instead using Armoni's generator \cite{arm98} (a generalization of the Nisan-Zuckerman generator \cite{nz96}), but even that slightly improved seed length is larger than the seed length of the generator we construct.
	
	\subsubsection{Finding noticeably large Fourier coefficients}
	Our randomness-efficient version of the Goldreich-Levin algorithm should be compared to the results of Bshouty et al.\ \cite{bjt04}, who gave several algorithms for finding noticeably large Fourier coefficients, all closely related to one another and based on an algorithm of Levin \cite{lev93}.
	\begin{itemize}
		\item Bshouty et al.\ gave one algorithm \cite[Figure 4]{bjt04} that makes $O(\frac{n}{\theta^2} \log(\frac{n}{\delta \theta}))$ queries and uses $O(n \log(\frac{n}{\theta}) \log(\frac{1}{\delta \theta}))$ random bits. Our algorithm has better randomness complexity, but worse query complexity.
		\item Bshouty et al.\ gave another algorithm \cite[Figure 5]{bjt04} that makes only $O(n/\theta^2)$ queries and uses just $O(\log(n/\theta) \cdot \log(1/\theta))$ random bits, but it merely outputs a list such that with probability $1/2$, some $U$ in the list satisfies $|\hat{F}(U)| \geq \theta$, assuming such a $U$ exists.
	\end{itemize}
	We also remark that there is a \emph{deterministic} version of the Goldreich-Levin algorithm for functions with bounded $\ell_1$ norm; this follows easily from the work of Kushilevitz and Mansour \cite{km93} (see also \cite[Section 6.4]{o'd14}). In contrast, our algorithm works for all functions $F: \{0, 1\}^n \to \{-1, 1\}$.
	
	\subsection{Outline of this paper}
	
	In \cref{sec:base}, we describe the shifting and rounding steward $\mathsf{S}_0$ and prove that it admits certification trees with a small branching factor. Then, in \cref{sec:prg}, we construct and analyze our pseudorandom generator for block decision trees, mimicking the construction and analysis of the INW generator. In \cref{sec:main}, we put these pieces together to prove our main result (\cref{thm:main}). In \cref{sec:variants}, we show how to construct our variant stewards. In \cref{sec:applications}, we explain our applications of our main steward. Finally, in \cref{sec:lower-bound}, we prove our randomness complexity lower bounds for stewards.
	
	\section{The shifting and rounding steward $\mathsf{S}_0$} \label{sec:base}
	
	As a building block for our main steward constructions, we first construct our randomness-inefficient one-query steward $\mathsf{S}_0$. Recall that any one-query steward makes two choices in each round: the input $X_i$ to $f_i$ and the estimate $Y_i \in \R^d$. The steward $\mathsf{S}_0$ focuses on the second choice: each $X_i$ is chosen uniformly at random, but $\mathsf{S}_0$ carefully shifts and rounds the output $f_i(X_i)$. (See \cref{fig:base-outline}.)
	
	\begin{figure}
		\centering
		\begin{framed}
			\begin{enumerate}
				\item For $i = 1$ to $k$:
				\begin{enumerate}
					\item $\mathsf{O}$ chooses $f_i: \{0, 1\}^n \to \R^d$ and gives it to $\mathsf{S}_0$.
					\item $\mathsf{S}_0$ picks \emph{fresh randomness} $X_i \in \{0, 1\}^n$ and queries to obtain $W_i \stackrel{\text{def}}{=} f_i(X_i)$.
					\item $\mathsf{S}_0$ computes $Y_i$ by shifting and rounding $W_i$ according to the algorithm in \cref{sec:base-construction}.
					\item $\mathsf{S}_0$ gives $Y_i$ to $\mathsf{O}$.
				\end{enumerate}
			\end{enumerate}
			\vspace{-0.4cm}
		\end{framed}
		\caption{Outline of $\mathsf{O} \leftrightarrow \mathsf{S}_0$.} \label{fig:base-outline}
	\end{figure}
	
	\subsection{The shifting and rounding algorithm} \label{sec:base-construction}
	We now describe the algorithm by which $\mathsf{S}_0$ computes $Y_i \in \R^d$ from $W_i \stackrel{\text{def}}{=} f_i(X_i)$. Fix $n, k, d \in \N$ and $\epsilon, \delta > 0$. Let $[d]$ denote the set $\{1, 2, \dots, d\}$. Partition $\R$ into half-open intervals of length $(d + 1) \cdot 2\epsilon$. Let $\mathcal{I}$ denote the set of these intervals. For $w \in \R$, let $\mathsf{Round}(w)$ denote the midpoint of the interval in $\mathcal{I}$ containing $w$. Given $W_i \in \R^d$:
	\begin{enumerate}
		\item Find $\Delta_i \in [d + 1]$ such that for every $j \in [d]$, there is some $I \in \mathcal{I}$ such that
		\[
			[W_{ij} + (2\Delta_i - 1)\epsilon, W_{ij} + (2\Delta_i + 1)\epsilon] \subseteq I.
		\]
		(We will show that such a $\Delta_i$ exists.)
		\item For every $j \in [d]$, set $Y_{ij} = \mathsf{Round}(W_{ij} + 2\Delta_i \epsilon)$.
	\end{enumerate}
	We must show that this algorithm is well-defined:
	\begin{lem} \label{lem:delta-exists}
		For any $W \in \R^d$, there exists $\Delta \in [d + 1]$ such that for every $j \in [d]$, there is a single interval in $\mathcal{I}$ that entirely contains $[W_j + (2\Delta - 1)\epsilon, W_j + (2\Delta + 1)\epsilon]$.
	\end{lem}
	
	\begin{proof}
		Consider picking $\Delta \in [d + 1]$ uniformly at random. For each $j$, the probability that two distinct intervals in $\mathcal{I}$ intersect $[W_j + (2\Delta - 1)\epsilon, W_j + (2\Delta + 1)\epsilon]$ is precisely $1/(d + 1)$ by our choice of the length of the intervals. The union bound over $d$ different $j$ values completes the proof.
	\end{proof}
	
	\subsection{Analysis: Certification trees} \label{sec:base-analysis}
	
	As outlined in \cref{sec:techniques-shifting-and-rounding}, the key lemma says that for any owner $\mathsf{O}$, there exists a block decision tree $T_{\mathsf{O}}$ with a small branching factor that certifies correctness of $\mathsf{O} \leftrightarrow \mathsf{S}_0$:
	\begin{lem} \label{lem:two-mediators}
		Assume $\delta < 1/2$. Let $\Sigma = [d + 1] \cup \{\bot\}$. For any deterministic owner $\mathsf{O}$, there exists a $(k, n, \Sigma)$ block decision tree $T_{\mathsf{O}}$ with the following properties.
		\begin{enumerate}
			\item For any internal node $v$, $\Pr_{X \in \{0, 1\}^n}[v(X) = \bot] \leq \delta$. \label{cond:low-failure-prob}
			\item Fix $X_1, \dots, X_k \in \{0, 1\}^n$, and suppose that the path from the root to $T_{\mathsf{O}}(X_1, \dots, X_k)$ does not include any $\bot$ nodes. Then $\max_i \|Y_i - \mu_i\|_{\infty} \leq O(\epsilon d)$ in $\mathsf{O} \leftrightarrow \mathsf{S}_0(X_1, \dots, X_k)$. \label{cond:model}
		\end{enumerate}
	\end{lem}
	
	Notice that \cref{lem:two-mediators} does not assert that $T_{\mathsf{O}}$ computes the transcript of $\mathsf{O} \leftrightarrow \mathsf{S}_0$. In fact, for the analysis, we will define \emph{another} steward $\mathsf{S}_0'$, and $T_{\mathsf{O}}$ will compute a sequence of values that arise in $\mathsf{O} \leftrightarrow \mathsf{S}_0'$. This new steward $\mathsf{S}_0'$ will be computationally inefficient; it will \emph{compress and decompress} the output of $\mathsf{S}_0$ (with some chance of failure) before giving it to $\mathsf{O}$, as we suggested in \cref{sec:techniques-shifting-and-rounding}.

	\begin{proof}[Proof of \cref{lem:two-mediators}]
		For a function $f: \{0, 1\}^n \to \R^d$ that is $(\epsilon, \delta)$-concentrated at some point $\mu \in \R^d$, define $\mu(f)$ to be the \emph{smallest} vector in $\R^d$ (under, say, the lexicographical order) at which $f$ is $(\epsilon, \delta)$-concentrated. (This exists, because $\{0, 1\}^n$ is finite, so the set of points where $f$ is concentrated is a compact subset of $\R^d$.) For a vector $Y \in \R^d$, say that a value $\Delta \in [d + 1]$ is \emph{$f$-compatible} with $Y$ if $Y_j = \mathsf{Round}(\mu(f)_j + 2 \Delta \epsilon)$ for every $j \in [d]$. Just for the analysis, let $\mathsf{S}_0'$ be the following (many-query) steward:
		\begin{enumerate}
			\item For $i = 1$ to $k$:
			\begin{enumerate}
				\item Give $f_i$ to $\mathsf{S}_0$, allowing it to make its one query and choose its output vector $Y_i \in \R^d$.
				\item Query $f_i$ at \emph{every} point in its domain, thereby learning the entire function.
				\item Compute
				\[
				\hat{\Delta}_i = \begin{cases}
				\text{the smallest } \Delta \in [d + 1] \text{ $f_i$-compatible with } Y_i & \text{if any such $\Delta$ exists} \\
				\bot & \text{otherwise.}
				\end{cases}
				\]
				\item Output $\hat{Y}_i = (\hat{Y}_{i1}, \dots, \hat{Y}_{id})$, where for each $j \in [d]$,
				\[
				\hat{Y}_{ij} = \begin{cases}
				\mathsf{Round}(\mu(f)_j + 2\hat{\Delta}_i \epsilon) & \text{if } \hat{\Delta}_i \neq \bot \\
				0 & \text{otherwise.}
				\end{cases}
				\]
			\end{enumerate}
		\end{enumerate}
		We are now ready to formally define $T_{\mathsf{O}}$ \emph{as a function}. Because $\mathsf{S}_0'(X_1, \dots, X_k)$ looks at $X_i$ only in round $i$, we can sensibly speak of the first $i$ rounds of $\mathsf{O} \leftrightarrow \mathsf{S}_0'(X_1, \dots, X_i)$ even for $i < k$. This allows us to define $T_{\mathsf{O}}(X_1, \dots, X_i)$ to be the node $v$ in $T_{\mathsf{O}}$ such that the path from the root to $v$ is described by the values $\hat{\Delta}_1, \dots, \hat{\Delta}_i$ that arise in $\mathsf{O} \leftrightarrow \mathsf{S}_0'(X_1, \dots, X_i)$.
		
		Now, we must show that this function $T_{\mathsf{O}}$ can be realized as a block decision tree, i.e. that each internal node $v$ can be assigned a transition function $v: \{0, 1\}^n \to \Sigma$ that is compatible with the definition of $T_{\mathsf{O}}$ as a function. Indeed, observe that $\hat{\Delta}_1, \dots, \hat{\Delta}_{i - 1}$ fully determine the state of $\mathsf{O}$ after the first $i - 1$ rounds of $\mathsf{O} \leftrightarrow \mathsf{S}_0'(X_1, \dots, X_i)$ and hence determine the function $f_i$. Furthermore, $\mathsf{S}_0$ is ``memoryless'', i.e. $Y_i$ is fully determined by $f_i$ and $X_i$. Thus, $\hat{\Delta}_i$ is fully determined by $\hat{\Delta}_1, \dots, \hat{\Delta}_{i - 1}$ and $X_i$. So there is a function $\varphi: (\hat{\Delta}_1, \dots, \hat{\Delta}_{i - 1}, X_i) \mapsto \hat{\Delta}_i$, and if the path from the root to $v$ is described by $\hat{\Delta}_1, \dots, \hat{\Delta}_{i - 1}$, we can set $v(X_i) \stackrel{\text{def}}{=} \varphi(\hat{\Delta}_1, \dots, \hat{\Delta}_{i - 1}, X_i)$.
		
		\paragraph{Analysis of $T_{\mathsf{O}}$} By the definition of $T_{\mathsf{O}}$ as a function, to prove \cref{cond:low-failure-prob} in the lemma statement, we must show that in each round of $\mathsf{O} \leftrightarrow \mathsf{S}_0'$, $\Pr[\hat{\Delta}_i = \bot] \leq \delta$. Indeed, by concentration, with probability $1 - \delta$, for every $j$, $|W_{ij} - \mu(f_i)_j| \leq \epsilon$. In this case, by the construction of $\mathsf{S}_0$, $W_{ij} + 2\Delta_i \epsilon$ and $\mu(f_i)_j + 2\Delta_i \epsilon$ are in the same interval in $\mathcal{I}$ for every $j \in [d]$. Therefore, in this case, there is at least one $\Delta$ value that is $f_i$-compatible with $Y_i$, namely the value $\Delta_i$ used by $\mathsf{S}_0$.
		
		Finally, to prove \cref{cond:model} in the lemma statement, suppose the path from the root node to $T_{\mathsf{O}}(X_1, \dots, X_k)$ does not include any $\bot$ nodes. Then in $\mathsf{O} \leftrightarrow \mathsf{S}_0'(X_1, \dots, X_k)$, for every $i$, $\hat{\Delta}_i \neq \bot$. This implies that every $Y_{ij}$ is of the form $\mathsf{Round}(\mu(f_i)_j + 2\hat{\Delta}_i \epsilon)$ for some $\hat{\Delta}_i \in [d + 1]$. Therefore, $|Y_{ij} - \mu(f_i)_j| \leq 3(d + 1) \epsilon$, since $2\hat{\Delta}_i \epsilon \leq 2(d + 1) \epsilon$ and rounding introduces at most $(d + 1) \epsilon$ additional error.
		
		Of course, so far the analysis has treated $\mathsf{S}'_0$, not $\mathsf{S}_0$. But the crucial point is, for every $i$, since $\hat{\Delta}_i \neq \bot$, we can be sure that $Y_i = \hat{Y}_i$. Therefore, the values $f_1, \dots, f_k, Y_1, \dots, Y_k$ in $\mathsf{O} \leftrightarrow \mathsf{S}_0'(X_1, \dots, X_k)$ are \emph{exactly the same} as they are in $\mathsf{O} \leftrightarrow \mathsf{S}_0(X_1, \dots, X_k)$! Therefore, in $\mathsf{O} \leftrightarrow \mathsf{S}_0(X_1, \dots, X_k)$, for every $i$, $\|Y_i - \mu(f_i)\|_{\infty} \leq (3d + 3) \epsilon$. Finally, since $\delta < 1/2$, if $\mu_i$ is \emph{any} point where $f_i$ is $(\epsilon, \delta)$-concentrated, $\|\mu(f_i) - \mu_i\|_{\infty} \leq 2\epsilon$. Therefore, for every $i$, $\|Y_i - \mu_i\|_{\infty} \leq 3(d + 1)\epsilon + 2\epsilon = (3d + 5)\epsilon$.
	\end{proof}
	
	Notice that in $\mathsf{O} \leftrightarrow \mathsf{S}_0'(X_1, \dots, X_k)$, if $\hat{\Delta}_i = \bot$ for some $i$, then the interaction might diverge from $\mathsf{O} \leftrightarrow \mathsf{S}_0(X_1, \dots, X_k)$, in which case $T_{\mathsf{O}}(X_1, \dots, X_k)$ does not encode the transcript of $\mathsf{O} \leftrightarrow \mathsf{S}_0(X_1, \dots, X_k)$ in any way.
	
	\section{Pseudorandomness for block decision trees} \label{sec:prg}
	
	Recall that our goal is to modify the internal parameters of the INW generator, thereby constructing a $\gamma$-PRG for $(k, n, \Sigma)$ block decision trees with seed length $n + O(k \log |\Sigma| + (\log k) \log(1/\gamma)$. The construction and analysis mimic the standard treatment of the INW generator, and the reader who is familiar with the INW generator is encouraged to skip to \cref{sec:prg-construction} to just see the new parameters. In words, the only new feature is that we use extractors for a geometrically growing entropy deficit at each level of the recursion to match the geometrically growing width of the block decision tree.
	
	\subsection{Formal definitions and theorem statement} \label{sec:prg-formalities}
	Let $U_n$ denote the uniform distribution on $\{0, 1\}^n$. For two probability distributions $\mu, \mu'$ on the same measurable space, write $\mu \sim_{\gamma} \mu'$ to indicate that $\mu$ and $\mu'$ have total variation distance at most $\gamma$.
	
	\begin{defn}
		We say that $\mathsf{Gen}: \{0, 1\}^s \to \{0, 1\}^{nk}$ is a \emph{$\gamma$-PRG for $(k, n, \Sigma)$ block decision trees} if for every such tree $T$, $T(\mathsf{Gen}(U_s)) \sim_{\gamma} T(U_{nk})$.
	\end{defn}
	
	\begin{thm} \label{thm:block-decision-tree-prg}
		For every $n, k \in \N$, every finite alphabet $\Sigma$, and every $\gamma > 0$, there exists a $\gamma$-PRG $\mathsf{Gen}: \{0, 1\}^s \to \{0, 1\}^{nk}$ for $(k, n, \Sigma)$ block decision trees with seed length
		\[
			s \leq n + O(k \log |\Sigma| + (\log k)\log(1/\gamma)).
		\]
		The PRG can be computed in $\poly(n, k, \log |\Sigma|, \log(1/\gamma))$ time.
	\end{thm}
	
	\subsection{Concatenating PRGs for block decision trees}
	Toward proving \cref{thm:block-decision-tree-prg}, for a $(k, n, \Sigma)$ block decision tree $T = (V, E)$ and a node $v \in V$, let $T_v$ denote the subtree rooted at $v$, and observe that we can think of $T_v$ as a $(k', n, \Sigma)$ block decision tree, where $k' = k - \text{depth}(v)$. This simple observation -- after a block decision tree has been computing for a while, the remaining computation is just another block decision tree -- implies that pseudorandom generators for block decision trees can be \emph{concatenated} with mild error accumulation. This fact and its easy proof are perfectly analogous to the situation with ordered branching programs. We record the details below.
	
	\begin{lem} \label{lem:concatenation}
		Suppose $\mathsf{Gen}_1: \{0, 1\}^{s_1} \to \{0, 1\}^{n k_1}$ is a $\gamma_1$-PRG for $(k_1, n, \Sigma)$ block decision trees and $\mathsf{Gen}_2: \{0, 1\}^{s_2} \to \{0, 1\}^{n k_2}$ is a $\gamma_2$-PRG for $(k_2, n, \Sigma)$ block decision trees. Let $\mathsf{Gen}(x, y) = (\mathsf{Gen}_1(x), \mathsf{Gen}_2(y))$. Then $\mathsf{Gen}$ is a $(\gamma_1 + \gamma_2)$-PRG for $(k_1 + k_2, n, \Sigma)$ block decision trees.
	\end{lem}
	
	\begin{proof}
		Fix a $(k_1 + k_2, n, \Sigma)$ block decision tree $T$. For a node $u$ at depth $k_1$ and a leaf node $v$, define
		\begin{align*}
			&p(u) = \Pr[T(U_{nk_1}) = u] 
			&&p(v \mid u) = \Pr[T_u(U_{n k_2}) = v] \\
			&\tilde{p}(u) = \Pr[T(\mathsf{Gen}_1(U_{s_1})) = u]
			&&\tilde{p}(v \mid u) = \Pr[T_u(\mathsf{Gen}_2(U_{s_2})) = v].
		\end{align*}
		To prove correctness of $\mathsf{Gen}$, recall that $\ell_1$ distance is twice total variation distance. The $\ell_1$ distance between $T(\mathsf{Gen}(U_{s_1 + s_2}))$ and $T(U_{n(k_1 + k_2)})$ is precisely $\sum_{u, v} |p(u) p(v \mid u) - \tilde{p}(u) \tilde{p}(v \mid u)|$. By the triangle inequality, this is bounded by
		\begin{align*}
			&\sum_{u, v} |p(u) p(v \mid u) - p(u) \tilde{p}(v \mid u)| + \sum_{u, v} |p(u) \tilde{p}(v \mid u) - \tilde{p}(u) \tilde{p}(v \mid u)| \\
			=\; & \sum_{u, v} p(u) \cdot |p(v \mid u) - \tilde{p}(v \mid u)| + \sum_{u, v} |p(u) - \tilde{p}(u)| \cdot \tilde{p}(v \mid u) \\
			=\; & \sum_u p(u) \sum_v |p(v \mid u) - \tilde{p}(v \mid u)| + \sum_u |p(u) - \tilde{p}(u)|.
		\end{align*}
		By the correctness of $\mathsf{Gen}_1$ and $\mathsf{Gen}_2$, this is bounded by $\left(\sum_u p(u) \cdot 2 \gamma_2\right) + 2 \gamma_1 = 2(\gamma_1 + \gamma_2)$.
	\end{proof}
	
	\subsection{Recycling randomness}
	We find it most enlightening to think of the INW generator in terms of extractors, as suggested by Raz and Reingold \cite{rr99} and in the spirit of the Nisan-Zuckerman generator \cite{nz96}. The analysis is particularly clean if we work with \emph{average-case extractors}, a concept introduced by Dodis et al.\ \cite{dors08}.
	\begin{defn}
		For discrete random variables $X, V$, the \emph{average-case conditional min-entropy} of $X$ given $V$ is
		\[
			\tilde{H}_{\infty}(X \mid V) = -\log_2\left(\E_{v \sim V}\left[2^{-H_{\infty}(X \mid V = v)}\right]\right),
		\]
		where $H_{\infty}$ is (standard) min-entropy.
	\end{defn}
	Intuitively, $\tilde{H}_{\infty}(X \mid V)$ measures the amount of randomness in $X$ from the perspective of someone who knows $V$. The output of an \emph{average-case extractor} is required to look uniform even from the perspective of someone who knows $V$, as long as its first input is sampled from a distribution that has high min-entropy conditioned on $V$:
	\begin{defn}
		We say that $\mathsf{Ext}: \{0, 1\}^s \times \{0, 1\}^d \to \{0, 1\}^m$ is an \emph{average-case $(s - t, \beta)$-extractor} if for every $X$ distributed on $\{0, 1\}^{s}$ and every discrete random variable $V$ such that $\tilde{H}_{\infty}(X \mid V) \geq s - t$, if we let $Y \sim U_d$ be independent of $(X, V)$ and let $Z \sim U_m$ be independent of $V$, then $(V, \mathsf{Ext}(X, Y)) \sim_{\beta} (V, Z)$.
	\end{defn}
	Average-case extractors are the perfect tools for \emph{recycling randomness} in space-bounded computation. We record the details for block decision trees below.
	\begin{lem}[Randomness recycling lemma for block decision trees] \label{lem:recycle}
		Suppose $\mathsf{Gen}: \{0, 1\}^s \to \{0, 1\}^{nk}$ is a $\gamma$-PRG for $(k, n, \Sigma)$ block decision trees and $\mathsf{Ext}: \{0, 1\}^s \times \{0, 1\}^d \to \{0, 1\}^s$ is an average-case $(s - k \log |\Sigma|, \beta)$-extractor. Define
		\[
			\mathsf{Gen}'(x, y) = (\mathsf{Gen}(x), \mathsf{Gen}(\mathsf{Ext}(x, y))).
		\]
		Then $\mathsf{Gen}'$ is a $(2\gamma + \beta)$-PRG for $(2k, n, \Sigma)$ block decision trees.
	\end{lem}
	
	\begin{proof}
		Let $T$ be a $(2k, n, \Sigma)$ block decision tree. Let $X \sim U_s$ and let $V = T(\mathsf{Gen}(X))$. By \cite[Lemma 2.2b]{dors08}, the fact that $V$ can be described using $k \log |\Sigma|$ bits implies that $\tilde{H}_{\infty}(X \mid V) \geq s - k \log |\Sigma|$. Therefore, by the average-case extractor condition, if we let $Y \sim U_d$ be independent of $X$ and $Z \sim U_d$ be independent of $V$, then
		\[
			(V, \mathsf{Ext}(X, Y)) \sim_{\beta} (V, Z).
		\]
		Applying a (deterministic) function can only make the distributions closer. Apply the function $(v, z) \mapsto T_v(\mathsf{Gen}(z))$:
		\[
			T(\mathsf{Gen}'(X, Y)) \sim_{\beta} T(\mathsf{Gen}(X), \mathsf{Gen}(Z)).
		\]
		By \cref{lem:concatenation}, the right-hand side is $(2\gamma)$-close to $T(U_{2nk})$. The triangle inequality completes the proof.
	\end{proof}
	
	To actually construct a generator, we will need to instantiate this randomness recycling lemma with an explicit average-case extractor:
	\begin{lem} \label{lem:explicit-extractor}
		For every $s, t \in \N$ and every $\beta > 0$, there exists an average-case $(s - t, \beta)$-extractor $\mathsf{Ext}: \{0, 1\}^s \times \{0, 1\}^d \to \{0, 1\}^s$ with seed length $d \leq O(t + \log(1/\beta))$ computable in time $\poly(s, \log(1/\beta))$.
	\end{lem}
	
	\begin{proof}[Proof sketch]
		It is standard (and can be proven using expanders, see, e.g., \cite{vad12}) that there exists an \emph{ordinary} $(s - t - \log(2/\beta), \beta/2)$-extractor $\mathsf{Ext}: \{0, 1\}^s \times \{0, 1\}^d \to \{0, 1\}^s$ with seed length $d \leq O(t + \log(1/\beta))$ computable in time $\poly(s, \log(1/\beta))$. By the same argument as that used to prove \cite[Lemma 2.3]{dors08}, $\mathsf{Ext}$ is automatically an average-case $(s - t, \beta)$-extractor.
	\end{proof}
	
	\subsection{The recursive construction} \label{sec:prg-construction}
	\begin{proof}[Proof of \cref{thm:block-decision-tree-prg}]
		Define $\beta = \gamma / 2^{\lceil \log k \rceil}$. For $i \geq 0$, define $s_i \in \N$, $d_i \in \N$, $G_i: \{0, 1\}^{s_i} \to \{0, 1\}^{n \cdot 2^i}$, and $\mathsf{Ext}_i: \{0, 1\}^{s_i} \times \{0, 1\}^{d_i} \to \{0, 1\}^{s_i}$ through mutual recursion as follows. Start with $s_0 = n$ and $G_0(x) = x$. Having already defined $s_i$ and $G_i$, let $\mathsf{Ext}_i$ be the average-case $(s_i - 2^i \log |\Sigma|, \beta)$-extractor of \cref{lem:explicit-extractor}, and let $d_i$ be its seed length. Then let $s_{i + 1} = s_i + d_i$, and let
		\[
			G_{i + 1}(x, y) = (G_i(x), G_i(\mathsf{Ext}_i(x, y))).
		\]
		We show by induction on $i$ that $G_i$ is a $(\beta \cdot (2^i - 1))$-PRG for $(2^i, n, \Sigma)$ block decision trees. In the base case $i = 0$, this is trivial. For the inductive step, apply \cref{lem:recycle}, and note that $2 \beta (2^i - 1) + \beta = \beta(2^{i + 1} - 1)$. This completes the induction. Therefore, we can let $\mathsf{Gen} = G_{\lceil \log k \rceil}$, since $\beta \cdot (2^{\lceil \log k \rceil} - 1) < \gamma$. The seed length $s_{\lceil \log k \rceil}$ of $\mathsf{Gen}$ is
		\begin{align*}
			n + \sum_{i = 0}^{\lceil \log k \rceil} d_i &\leq n + O\left(\sum_{i = 0}^{\lceil \log k \rceil} (2^i \log |\Sigma| + \log k + \log(1/\gamma))\right) \\
			 &\leq n + O(k \log |\Sigma| + (\log k) \log(1/\gamma)).
		\end{align*}
		The time needed to compute $\mathsf{Gen}(x)$ is just the time needed for $O(k)$ applications of $\mathsf{Ext}_i$ for various $i \leq O(\log k)$, which is $\poly(n, k, \log |\Sigma|, \log(1/\gamma))$.
	\end{proof}
	
	\section{Proof of main result (\cref{thm:main})} \label{sec:main}
	
	Without loss of generality, assume $\delta < 1/2$. (If $\delta \geq 1/2$, then either $k = 1$ or $k \delta \geq 1$; in either case, the result is trivial.) Let $\mathsf{S}_0$ be the steward of \cref{sec:base}, let $\Sigma$ be the alphabet of \cref{lem:two-mediators}, and let $\mathsf{Gen}$ be the $\gamma$-PRG for $(k, n, \Sigma)$ block decision trees of \cref{thm:block-decision-tree-prg}. The steward is $\mathsf{S}(X) \stackrel{\text{def}}{=} \mathsf{S}_0(\mathsf{Gen}(X))$.
	
	Consider any owner $\mathsf{O}$. We may assume without loss of generality that $\mathsf{O}$ is deterministic, because a randomized owner is just a distribution over deterministic owners. By \cref{cond:low-failure-prob} of \cref{lem:two-mediators} and the union bound,
	\[
		\Pr[\text{some node in the path from the root to } T_{\mathsf{O}}(U_{nk}) \text{ is labeled } \bot] \leq k\delta.
	\]
	Therefore, when $T_{\mathsf{O}}$ reads $\mathsf{Gen}(U_s)$ instead of $U_{nk}$, the probability is at most $k\delta + \gamma$. By \cref{cond:model} of \cref{lem:two-mediators}, this proves the correctness of $\mathsf{S}$. The randomness complexity of $\mathsf{S}$ is just the seed length of $\mathsf{Gen}$, which is indeed $n + O(k \log |\Sigma| + (\log k) \log(1/\gamma)) = n + O(k \log(d + 1) + (\log k) \log(1/\gamma))$. The total runtime of $\mathsf{S}$ is clearly $\poly(n, k, d, \log(1/\epsilon), \log(1/\gamma))$.\footnote{We assume here that our computational model allows the necessary arithmetic and rounding of \cref{sec:base-construction} to be performed efficiently, even if the owner chooses an $f_i$ that outputs vectors whose coordinates are very large numbers.} \qed
	
	\section{Variant stewards} \label{sec:variants}
	
	\begin{thm} \label{thm:computationally-inefficient}
		For any $n, k, d \in \N$, for any $\epsilon, \delta, \gamma > 0$, there exists a (computationally inefficient) one-query $(O(\epsilon d), k\delta + \gamma)$-steward for $k$ adaptively chosen $(\epsilon, \delta)$-concentrated functions $f_1, \dots, f_k: \{0, 1\}^n \to \R^d$ with randomness complexity
		\[
		n + k \log(d + 2) + 2\log(1/\gamma) + O(1).
		\]
	\end{thm}
	\begin{proof}[Proof sketch]
		Mimic the proof of \cref{thm:main}, but use a PRG obtained by the standard nonconstructive argument (\cref{apx:prg}).
	\end{proof}
	
	The shifting and rounding steward $\mathsf{S}_0$ can be generalized to achieve a tradeoff between low error $\epsilon'$ and low branching factor $|\Sigma|$ of the certification tree $T_{\mathsf{O}}$. In particular, for any factorization $d = d_0 d_1$, one can reduce the error from $O(\epsilon d)$ down to $O(\epsilon d_0)$ at the cost of increasing the branching factor of $T_{\mathsf{O}}$ from $d + 2$ up to $(d_0 + 1)^{d_1} + 1$. This is achieved by simply partitioning the $d$ coordinates into $d_1$ groups of $d_0$ coordinates and shifting each group individually; the details are in \cref{apx:generalized-shifting-and-rounding}. This immediately implies the following generalization of \cref{thm:main}, which achieves a tradeoff between error and randomness complexity:
	
	\begin{thm} \label{thm:prg-steward}
		For any $n, k, d, d_0 \in \N$ with $d_0 \leq d$, for any $\epsilon, \delta, \gamma > 0$, there exists a one-query $(O(\epsilon d_0), k\delta + \gamma)$-steward for $k$ adaptively chosen $(\epsilon, \delta)$-concentrated functions $f_1, \dots, f_k: \{0, 1\}^n \to \R^d$ with randomness complexity
		\[
			n + O\left(\frac{kd \log(d_0 + 1)}{d_0} + (\log k) \log(1/\gamma)\right).
		\]
		The total running time of the steward is $\poly(n, k, d, \log(1/\epsilon), \log(1/\gamma))$.
	\end{thm}
	
	Recall from the introduction that if $f_1, \dots, f_k$ are chosen nonadaptively, then we can reuse randomness and just union bound over the $k$ functions. We now show that we can reuse the randomness in $\mathsf{S}_0$, as long as we union bound over \emph{all the nodes} in the certification tree. (This is similar to the analysis of the Saks-Zhou steward, except that in the Saks-Zhou case, the branching factor of the tree is just $1$. It is also similar to the analysis in \cite{bh15}.) This gives a steward with very low randomness complexity but large failure probability:
	
	\begin{thm} \label{thm:union-bound}
		For any $n, k, d, d_0 \in \N$ with $d_0 \leq d$, for any $\epsilon, \delta > 0$, there exists a one-query $(O(\epsilon d_0), \delta')$-steward for $k$ adaptively chosen $(\epsilon, \delta)$-concentrated functions $f_1, \dots, f_k: \{0, 1\}^n \to \R^d$ with randomness complexity $n$, where
		\[
		\delta' \leq \exp\left(O\left(\frac{kd \log(d_0 + 1)}{d_0}\right)\right) \cdot \delta.
		\]
		The total running time of the steward is $\poly(n, k, d, \log(1/\epsilon))$.
	\end{thm}
	
	\begin{proof}
		Assume without loss of generality that $d$ is a multiple of $d_0$ and that $\delta < 1/2$. The steward is $\mathsf{S}(X) \stackrel{\text{def}}{=} \mathsf{S}_0(X, X, X, \dots, X)$, where $\mathsf{S}_0$ is the steward of \cref{sec:base} generalized as in \cref{apx:generalized-shifting-and-rounding}. To prove correctness, fix any deterministic owner $\mathsf{O}$. Let $T_{\mathsf{O}}$ be the block decision tree of \cref{lem:generalized-two-mediators}. By \cref{cond:generalized-low-failure-prob} of \cref{lem:generalized-two-mediators}, from any internal node, if $T_{\mathsf{O}}$ reads $X$, the probability that it moves to the $\bot$ child is at most $\delta$. Therefore, by the union bound over all nodes, the probability that there is some node from which $T_{\mathsf{O}}$ would move to the $\bot$ child upon reading $X$ is at most the value $\delta'$ in the lemma statement. By \cref{cond:generalized-model} of \cref{lem:generalized-two-mediators}, if no node in $T_{\mathsf{O}}$ takes a $\bot$ transition upon reading $X$, then  $\max_i \|\mu_i - Y_i\|_{\infty} \leq O(\epsilon d_0)$ in $\mathsf{O} \leftrightarrow \mathsf{S}(X)$.
	\end{proof}
	
	\section{Applications} \label{sec:applications}
	\subsection{Acceptance probabilities of Boolean circuits} \label{sec:circuit}
	
	A \emph{$(\epsilon, \delta)$-sampler} for Boolean functions on $n$ bits is a randomized oracle algorithm $\mathsf{Samp}$ such that for any Boolean function $C: \{0, 1\}^n \to \{0, 1\}$, if we let $\mu(C) \stackrel{\text{def}}{=} 2^{-n} \sum_x C(x)$, then
	\[
		\Pr[|\mathsf{Samp}^C - \mu(C)| > \epsilon] \leq \delta.
	\]
	We will use a near-optimal sampler constructed by Goldreich and Wigderson \cite{gw97}:
	\begin{lem}[{\cite[Theorem 6.5]{gw97}}] \label{lem:gw-boolean-sampler}
		For every $n \in \N$ and every $\epsilon, \delta > 0$, there is an $(\epsilon, \delta)$-sampler for Boolean functions on $n$ bits that makes $O(\log(1/\delta)/\epsilon^2)$ queries, uses $n + O(\log(1/\delta))$ random bits, and runs in time $\poly(n, 1/\epsilon, \log(1/\delta))$.
	\end{lem}
	
	\begin{proof}[Proof of \cref{cor:circuit}]
		Let $c$ be the constant under the $O(\cdot)$ of the error $\epsilon'$ in the steward of \cref{thm:main}. When given parameters $n, k, \epsilon, \delta$, let $\mathsf{Samp}$ be the Boolean $(\epsilon/c, \delta/(2k))$-sampler of \cref{lem:gw-boolean-sampler}, and say it uses $m$ coins. Let $\mathsf{S}$ be the $(\epsilon, \delta)$-steward of \cref{thm:main} for $k$ adaptively chosen $(\epsilon/c, \delta/(2k))$-concentrated functions $f_1, \dots, f_k: \{0, 1\}^m \to \R$. (So $\gamma = \delta/2$.) When given circuit $C_i$, define $f_i(X) = \mathsf{Samp}^{C_i}(X)$, i.e. the output $\mathsf{Samp}^{C_i}$ with randomness $X$. Give $f_i$ to $\mathsf{S}$, and output the value $Y_i$ that it returns.
		
		Proof of correctness: The definition of a sampler implies that each $f_i$ is $(\epsilon/c, \delta/(2k))$-concentrated at $\mu(C_i)$. Furthermore, each $f_i$ is defined purely in terms of $C_i$, which is chosen based only on $Y_1, \dots, Y_{i - 1}$. Therefore, the steward guarantee implies that with probability $1 - \delta$, every $Y_i$ is within $\pm \epsilon$ of $\mu(C_i)$.
		
		Randomness complexity analysis: The number of bits $m$ used by the sampler is $n + O(\log(k/\delta))$. Therefore, the number of bits used by the steward is
		\[
			n + O(\log(k/\delta)) + O(k + (\log k)\log(1/\delta)) = n + O(k + (\log k)\log(1/\delta)).
		\]
		
		Runtime analysis: The runtime of the steward is $\poly(m, k, \log(1/\gamma)) = \poly(n, k, \log(1/\delta))$. The runtime of the sampler is $\poly(n, 1/\epsilon, \log k, \log(1/\delta))$. The time required to evaluate each query of the sampler in round $i$ is $O(\text{size}(C_i))$ (assuming we work with a suitable computational model and a suitable encoding of Boolean circuits.) The number of queries that the sampler makes in each round is $O(\log(k/\delta)/\epsilon^2)$. Therefore, the total runtime of this algorithm is
		\[
			O\left(\frac{\log k + \log(1/\delta)}{\epsilon^2} \cdot \sum_{i = 1}^k \text{size}(C_i)\right) + \poly(n, k, 1/\epsilon, \log(1/\delta)). \qedhere
		\]
	\end{proof}
	
	\subsection{Simulating a $\mathbf{promise\mhyphen BPP}$ oracle} \label{sec:promise-bpp}
	
	\begin{thm} \label{thm:promise-bpp}
		Suppose a search problem $\Pi$ can be solved by a deterministic $\mathbf{promise\mhyphen BPP}$-oracle algorithm that runs in time $T$ and makes $k$ queries, and suppose that (regardless of previous oracle responses) each query of this algorithm can be decided by a randomized algorithm that runs in time $T'$, uses $n$ coins, and has failure probability $1/3$. Then for any $\delta$, $\Pi$ can be solved by a randomized (non-oracle) algorithm that runs in time
		\[
			T + O(T' \cdot k \log(k/\delta)) + \poly(n, k, \log(1/\delta)),
		\]
		has randomness complexity
		\[
			n + O(k + (\log k) \log(1/\delta)),
		\]
		and has failure probability $\delta$.
	\end{thm}
	
	(Recall that search problems generalize decision problems and function problems. In reality, the theorem generalizes to just about any kind of ``problem'', but we restrict ourselves to search problems for concreteness.) The theorem can easily be extended to randomized oracle algorithms by considering the problem of executing the randomized oracle algorithm using a given randomness string.
	
	As a reminder, as discussed in \cref{sec:pseudodeterminism}, \cref{thm:promise-bpp} would be trivial if it involved a $\mathbf{BPP}$ oracle instead of a $\mathbf{promise\mhyphen BPP}$ oracle. Indeed, in the $\mathbf{BPP}$ case, the randomness can be reduced to just $n + O(\log k + \log(1/\delta))$. This is because a $\mathbf{BPP}$ algorithm is pseudodeterministic, so the randomness can be safely reused from one query to the next. A $\mathbf{promise\mhyphen BPP}$ algorithm is not pseudodeterministic in general -- it is only guaranteed to be pseudodeterministic on inputs that satisfy the promise. 
	
	\begin{proof}[Proof sketch of \cref{thm:promise-bpp}]
		Let $\mathsf{B}$ be the algorithm of \cref{cor:circuit} with $\epsilon = 1/10$ and the desired failure probability $\delta$. When the oracle algorithm makes query $i$, define $f_i(X)$ to be the value outputted by the $\mathbf{promise\mhyphen BPP}$ algorithm on that query string using randomness $X$. Give $\mathsf{B}$ the ``circuit'' $f_i$. (The algorithm $\mathsf{B}$ treats the circuits as black boxes, so we don't need to bother implementing $f_i$ as a literal Boolean circuit; the important thing is that $f_i(X)$ can be evaluated in time $T'$.) When $\mathsf{B}$ outputs a value $Y_i$, give the oracle algorithm the response $0$ if $Y_i < 1/2$ and $1$ if $Y_i \geq 1/2$.
	\end{proof}
	
	\subsection{Simulating an $\mathbf{APP}$ oracle} \label{sec:app}
	
	Following Moser \cite{mos01}, we model oracle access to $\varphi \in \mathbf{APP}$ by requiring the oracle algorithm to provide $w \in \{0, 1\}^n$ and a unary representation of $1/\epsilon \in \N$; the oracle is guaranteed to respond with a value that is within $\pm \epsilon$ of $\varphi(w)$. From these definitions, it is easy to show, for example, that $\mathbf{BPP}^{\mathbf{APP}} = \mathbf{BPP}$. Just like we did with $\mathbf{promise\mhyphen BPP}$, we now use our steward to construct a time- and randomness-efficient simulation of any algorithm with an oracle for $\mathbf{APP}$.
	
	\begin{thm} \label{thm:app}
		Suppose $\varphi \in \mathbf{APP}$ and a search problem $\Pi$ can be solved by a deterministic $\varphi$-oracle algorithm that runs in time $T$ and makes $k$ queries $(w_1, \epsilon), \dots, (w_k, \epsilon)$ (where $w_i$ depends on previous oracle responses, but $\epsilon$ is the same for every query.) Let $c$ be the constant under the $O(\cdot)$ in the error $\epsilon'$ in \cref{thm:main}. Suppose that (regardless of the oracle responses) $\varphi(w_i)$ can be approximated to within $\pm \epsilon/c$ by a randomized algorithm that runs in time $T'$, uses $n$ coins, and has failure probability $1/3$. Then for any $\delta$, $\Pi$ can be solved by a randomized (non-oracle) algorithm that runs in time
		\[
			T + O(T' \cdot k \log(k/\delta)) + \poly(n, k, \log(1/\delta)),
		\]
		has randomness complexity
		\[
			n + O(k + (\log k) \log(1/\delta)),
		\]
		and has failure probability $\delta$.
	\end{thm}
	
	The proof of \cref{thm:app} is similar to the proofs of \cref{cor:circuit} and \cref{thm:promise-bpp}. The difference is that a sampler as defined previously is no longer quite the right tool for deterministic amplification; to amplify an $\mathbf{APP}$ algorithm, we are not trying to estimate the \emph{acceptance probability} of a Boolean function, but rather the point where a $[0, 1]$-valued function is \emph{concentrated}. For this, we use an \emph{averaging sampler}.
	
	An \emph{averaging $(\epsilon, \delta)$-sampler} for Boolean functions on $n$ bits is an algorithm $\mathsf{Samp}: \{0, 1\}^m \to (\{0, 1\}^n)^t$ such that for any Boolean function $C: \{0, 1\}^n \to \{0, 1\}$, if we let $\mu(C) \stackrel{\text{def}}{=} 2^{-n} \sum_x C(x)$ be the acceptance probability of $C$, then
	\[
		\Pr_{X \in \{0, 1\}^m}\left[\left|\mu(C) - \frac{1}{t} \sum_{i = 1}^t C(\mathsf{Samp}(X)_i)\right| > \epsilon\right] \leq \delta.
	\]
	(Note that an averaging sampler induces a sampler of a very specific form: query the oracle at several points and output the empirical mean.) We now show that an averaging sampler can be used to decrease the failure probability of a concentrated function by taking a \emph{median}. This observation (in a different form) is due to Bellare, Goldreich, and Goldwasser \cite{bgg93}.
	
	\begin{lem} \label{lem:app-amplification}
		Suppose $f: \{0, 1\}^n \to \R$ is $(\epsilon, \delta_0)$-concentrated at $\mu \in \R$ and $\mathsf{Samp}: \{0, 1\}^m \to (\{0, 1\}^n)^t$ is an averaging $(\epsilon', \delta)$-sampler for Boolean functions on $n$ bits, where $\epsilon' + \delta_0 < 1/2$. Define $g: \{0, 1\}^m \to \R$ by
		\[
			g(x) = \median_{i \in [t]} f(\mathsf{Samp}(x)_i).
		\]
		Then $g$ is $(\epsilon, \delta)$-concentrated at $\mu$.
	\end{lem}
	
	\begin{proof}
		Let $C: \{0, 1\}^n \to \{0, 1\}$ be the indicator function for $\{x : |f(x) - \mu| \leq \epsilon\}$. Then by the concentration of $f$, $2^{-n} \sum_x C(x) \geq 1 - \delta_0$. Therefore, by the averaging sampler condition, with probability $1 - \delta$ over $x$, $\frac{1}{t} \sum_i C(\mathsf{Samp}(X)_i) \geq 1 - \delta_0 - \epsilon' > 1/2$. If this is the case, then more than half of the values $f(\mathsf{Samp}(x)_1), \dots, f(\mathsf{Samp}(x)_t)$ are within $\pm \epsilon$ of $\mu$, which implies that their median is within $\pm \epsilon$ of $\mu$.
	\end{proof}
	
	The following lemma gives the parameters achieved by the famous ``random walk on expanders'' averaging sampler; see, e.g., \cite[Corollary 4.41]{vad12}.
	
	\begin{lem} \label{lem:averaging-sampler}
		For every $n \in \N$ and every $\epsilon, \delta > 0$, there is an averaging $(\epsilon, \delta)$-sampler for Boolean functions on $n$ bits with $m \leq n + O(\log(1/\delta)/\epsilon^2)$ and $t \leq O(\log(1/\delta)/\epsilon^2)$, computable in time $\poly(n, 1/\epsilon, \log(1/\delta))$.
	\end{lem}
	
	\begin{cor}[Deterministic amplification for $\mathbf{APP}$] \label{cor:app-amplification}
		Suppose $\varphi \in \mathbf{APP}$ via an algorithm that on input $(w, \epsilon)$ uses $n$ coins and $t$ time steps to compute $\varphi(w) \pm \epsilon$ with failure probability $1/3$. Then for any $\delta$, is possible to compute $\varphi(w) \pm \epsilon$ with failure probability $\delta$ using $O(t \log(1/\delta)) + \poly(n, \log(1/\delta)$ time steps and $n + O(\log(1/\delta))$ coins.
	\end{cor}
	
	\begin{proof}
		On input $(w, \epsilon)$:
		\begin{enumerate}
			\item Let $\mathsf{Samp}: \{0, 1\}^m \to (\{0, 1\}^n)^t$ be the averaging $(1/10, \delta)$-sampler for Boolean functions on $n$ bits of \cref{lem:averaging-sampler}.
			\item Define $f: \{0, 1\}^n \to [0, 1]$ by letting $f(X)$ be the output of the $1/3$-error-probability algorithm for computing $\varphi(w) \pm \epsilon$ on randomness $X$.
			\item Pick $X \in \{0, 1\}^m$ uniformly at random and return $\median_{i \in [t]} f(\mathsf{Samp}(X)_i)$.
		\end{enumerate}
		Correctness follows immediately from \cref{lem:app-amplification}, since $f$ is $(\epsilon, 1/3)$-concentrated at $\varphi(w)$. Efficiency follows immediately from \cref{lem:averaging-sampler}.
	\end{proof}
	
	\begin{proof}[Proof of \cref{thm:app}]
		By \cref{cor:app-amplification}, there is an algorithm $\Phi$ for computing $\varphi(w_i) \pm \epsilon/c$ with failure probability $\delta/(2k)$ that runs in time $O(T' \cdot \log(k/\delta)) + \poly(n, \log k, \log(1/\delta))$ and uses $m \leq n + O(\log(k/\delta))$ coins. Let $\mathsf{S}$ be the $(\epsilon, \delta)$-steward of \cref{thm:main} for $k$ adaptively chosen $(\epsilon/c, \delta/(2k))$-concentrated functions $f_1, \dots, f_k: \{0, 1\}^m \to \R$. (So $\gamma = \delta/2$.) When the oracle algorithm makes query $i$ about string $w_i$, let $f_i(X) = \Phi(w_i, \epsilon/c, X)$ and give $f_i$ to $\mathsf{S}$. When $\mathsf{S}$ outputs a value $Y_i$, give it to the oracle algorithm.
		
		Proof of correctness: Each $f_i$ is $(\epsilon/c, \delta/(2k))$-concentrated at $\varphi(w_i)$. Furthermore, each $f_i$ depends only on the previous oracle responses, i.e. $Y_1, \dots, Y_{i - 1}$. Therefore, the steward guarantee implies that with probability $1 - \delta$, every $Y_i$ is within $\pm \epsilon$ of $\varphi(w_i)$. If this occurs, then the oracle algorithm is guaranteed to give a correct output.
		
		Randomness complexity analysis: The number of bits used by the steward is
		\[
			m + O(k + (\log k) \log(1/\delta)) = n + O(k + (\log k) \log(1/\delta)).
		\]
		Runtime analysis: The runtime of the steward is $\poly(m, k, \log(1/\gamma)) = \poly(n, k, \log(1/\delta))$. Therefore, the total runtime is bounded by
		\[
			T + k \cdot (O(T' \cdot \log(k/\delta)) + \poly(n, \log k, \log(1/\delta))) + \poly(n, k, \log(1/\delta)),
		\]
		which is bounded by the expression in the theorem statement.
	\end{proof}
	
	\subsection{The Goldreich-Levin algorithm} \label{sec:goldreich-levin}
	
	\begin{thm}[Randomness-efficient Goldreich-Levin algorithm] \label{thm:gl}
		There is a randomized algorithm that, given oracle access to $F: \{0, 1\}^n \to \{-1, 1\}$ and given input parameters $\delta, \theta > 0$, outputs a list $L$ of subsets of $[n]$ such that with probability $1 - \delta$,
		\begin{enumerate}
			\item every $U$ satisfying $|\hat{F}(U)| \geq \theta$ is in $L$, and
			\item every $U \in L$ satisfies $|\hat{F}(U)| \geq \theta/2$.
		\end{enumerate}
		The number of queries made by the algorithm is
		\[
			O\left(\frac{n}{\theta^{11} \log(1/\theta)} \log\left(\frac{n}{\delta \theta}\right)\right),
		\]
		the number of random bits used by the algorithm is
		\[
			O(n + (\log n) \log(1/\delta)),
		\]
		and the runtime of the algorithm is $\poly(n, 1/\theta, \log(1/\delta))$.
	\end{thm}
	
	For comparison, using standard techniques (the GW sampler, reusing randomness within each round of adaptivity), the Goldreich-Levin algorithm can be implemented in a straightforward way to use $O(\frac{n}{\theta^6} \log(\frac{n}{\delta \theta}))$ queries and $O(n^2 + n \log(\frac{n}{\delta \theta}))$ random bits. So our algorithm significantly improves the randomness complexity at the expense of substantially increasing the exponent of $1/\theta$ in the query complexity.
	
	Toward proving \cref{thm:gl}, for a string $x \in \{0, 1\}^{\leq n}$, define
	\[
		\mathcal{U}(x) = \{U \subseteq [n]: \forall j \leq |x|, j \in U \iff x_j = 1\}.
	\]
	(That is, we think of $x \in \{0, 1\}^{\ell}$ as specifying $U \cap [\ell]$ in the natural way.) Define $W_x[F] = \sum_{U \in \mathcal{U}(x)} \hat{F}(U)^2$. One of the key facts used in the standard Goldreich-Levin algorithm is that $W_x[F]$ can be estimated using few queries to $F$; here, we use the GW sampler to improve the randomness efficiency of that estimation.
	\begin{lem} \label{lem:estimate-weight}
		There is a randomized algorithm that, given oracle access to $F$ and inputs $x \in \{0, 1\}^{\leq n}$, $\epsilon, \delta > 0$, estimates $W_x[F]$ to within $\pm \epsilon$ with failure probability $\delta$. The number of queries is $O(\log(1/\delta)/\epsilon^2)$, the number of random bits is $O(n + \log(1/\delta))$, and the runtime is $\poly(n, 1/\epsilon, \log(1/\delta))$.
	\end{lem}
	
	\begin{proof}
		Let $\ell = |x|$. As shown in the proof of \cite[Proposition 3.40]{o'd14},
		\[
		W_x[F] = \E_{\substack{y, y' \in \{0, 1\}^{\ell} \\ z \in \{0, 1\}^{n - \ell}}}[F(y, z) \cdot F(y', z) \cdot \chi_x(y) \cdot \chi_x(y')],
		\]
		where $\chi_x(y) \stackrel{\text{def}}{=} \prod_{j : x_j = 1} (-1)^{y_j}$. Let $C: \{0, 1\}^{n + \ell} \to \{0, 1\}$ be the function
		\[
		C(y, y', z) = \frac{1}{2} + \frac{1}{2} \cdot F(y, z) \cdot F(y', z) \cdot \chi_x(y) \cdot \chi_x(y'),
		\]
		so that $W_x[F] = 2 \E_{y, y', z}[C(y, y', z)] - 1$. We can estimate the expectation of $C$ to within $\pm \epsilon/2$ with failure probability $\delta$ using the GW sampler of \cref{lem:gw-boolean-sampler}, which implies an estimate of $W_x[F]$ to within $\pm \epsilon$. The number of queries made by the GW sampler is $O(\log(1/\delta)/\epsilon^2)$, and each query to $C$ can be evaluated by making $2$ queries to $F$. The randomness complexity of the GW sampler is $n + \ell + O(\log(1/\delta))$, which is $O(n + \log(1/\delta))$.
	\end{proof}
	
	The standard Goldreich-Levin algorithm proceeds by finding, for $\ell = 1$ to $n$, the set of all $x$ with $|x| = \ell$ such that $W_x[F] \gtrsim \theta^2$. In each round, the algorithm estimates $W_x[F]$ for all strings $x$ formed by appending a single bit to a string $x'$ that was previously found to satisfy $W_{x'}[F] \gtrsim \theta^2$. This adaptive structure is exactly suited for saving random bits using a steward. To further drive down the randomness complexity, we reduce the number of rounds of adaptivity by appending $\log(1/\theta)$ bits at a time instead of $1$ bit.
	
	\begin{proof}[Proof of \cref{thm:gl}]
		Algorithm:
		\begin{enumerate}
			\item Let $u = \lfloor \log(1/\theta) \rfloor$, let $k = \lceil n / u \rceil$, and let $d = \lfloor 2^u \cdot 4/\theta^2 \rfloor$.
			\item Let $\mathsf{S}$ be a $(\theta^2/4, \delta)$-steward for $k$ adaptively chosen $(\epsilon, \delta/(2n))$-concentrated functions $f_1, \dots, f_k: \{0, 1\}^m \to \R^d$, where $\epsilon \geq \Omega(\theta^2/d)$ and $m$ will become clear later.
			\item Set $L_0 := \{\text{empty string}\}$.
			\item For $i = 1$ to $k$:
			\begin{enumerate}
				\item If $|L_{i - 1}| > d/2^u$, abort and output ``fail''.
				\item Observe that every string in $L_{i - 1}$ has length $\ell = u(i - 1) < n$. Let $x_1, \dots, x_t$ be the set of all strings obtained from strings in $L_{i - 1}$ by appending $\min\{u, n - \ell\}$ bits, so $t \leq 2^u|L_{i - 1}| \leq d$.
				\item Define $f_i: \{0, 1\}^m \to \R^t$ by letting $f_i(X)_j$ be the estimate of $W_{x_j}[F]$ to within $\pm \epsilon$ provided by the algorithm of \cref{lem:estimate-weight} with failure probability $\delta/(2dn)$ using randomness $X$. Observe that by the union bound, $f_i$ is $(\epsilon, \delta/(2n))$-concentrated at $(W_{x_1}[F], \dots, W_{x_t}[F])$.
				\item By giving $f_i$ to $\mathsf{S}$, obtain estimates $\mu_1, \dots, \mu_t$ for $W_{x_1}[F], \dots, W_{x_t}[F]$.
				\item Set $L_i := \{x_j : \mu_j \geq \theta^2/2\}$.
			\end{enumerate}
			\item Output $L \stackrel{\text{def}}{=} \bigcup_{x \in L_k} \mathcal{U}(x)$.
		\end{enumerate}
		As hopefully became clear, $m$ is the number of random bits used by the algorithm of \cref{lem:estimate-weight}. With probability $1 - \delta$, all of the responses of $\mathsf{S}$ are accurate, i.e. every $\mu_j$ value is within $\pm \theta^2/4$ of the corresponding $W_{x_j}[F]$ value. Assume from now on that this has happened.
		
		By the definition of $L_i$, every $x$ in every $L_i$ satisfies $W_x[F] \geq \theta^2/4$. By Parseval's theorem (see, e.g., \cite[Section 1.4]{o'd14}), this implies that $|L_i| \leq 4/\theta^2 \leq d/2^u$ for every $i$. Therefore, the algorithm does not abort. Let $\ell_i$ be the length of all the strings in $L_i$, so $\ell_i = u i$ for $i < k$ and $\ell_k = n$. Suppose $\hat{F}(U)^2 \geq \theta^2$. By induction on $i$, the unique string $x \in \{0, 1\}^{\ell_i}$ with $U \in \mathcal{U}(x)$ is placed in $L_i$, because the estimate of $W_x[F]$ is at least $3\theta^2/4 > \theta^2/2$. This shows that $U \in L$. Conversely, if $U$ ends up in $L$, then the estimate of $\hat{F}(U)^2$ in iteration $i = n$ was at least $\theta^2/2$, so $\hat{F}(U)^2 \geq \theta^2/4$. This completes the proof of correctness of the algorithm.
		
		Now, observe that the total number of queries to $F$ is at most $kd$ times the $O(\log(nd/\delta)/\epsilon^2)$ queries that the algorithm of \cref{lem:estimate-weight} makes, i.e. the total number of queries to $F$ is
		\[
			O\left(\frac{kd^3 \log(nd/\delta)}{\theta^2}\right) = O\left(\frac{n}{\theta^{11}\log(1/\theta)} \log\left(\frac{n}{\delta \theta}\right)\right).
		\]
		The randomness complexity of the algorithm is just the randomness complexity of $\mathsf{S}$. We will use the steward of \cref{thm:main} with $\gamma = \delta/2$, so the randomness complexity is $m + O(k \log(d + 1) + (\log k)\log(1/\delta))$. Since $m \leq O(n + \log(n/(\delta \theta)))$, the total randomness complexity is
		\[
			O\left(n + \frac{n}{\log(1/\theta)} \log(1/\theta) + (\log n) \log(1/\delta) + \log(1/\theta)\right) = O(n + (\log n) \log(1/\delta) + \log(1/\theta)).
		\]
		To get rid of the $\log(1/\theta)$ term as claimed in the theorem statement, just notice that we can assume without loss of generality that $\theta \geq 2^{-n + 1}$, because any nonzero Fourier coefficient of a $\{-1, 1\}$-valued function has absolute value at least $2^{-n + 1}$. The total runtime of the algorithm is clearly $\poly(n, 1/\theta, \log(1/\delta))$.
	\end{proof}
	
	\section{Randomness complexity lower bounds} \label{sec:lower-bound}
	
	\subsection{Lower bound for pseudorandom generation stewards} \label{sec:prg-steward}
	
	\begin{defn}
		A \emph{pseudorandom generation steward} is a one-query steward with the additional property that in each round $i$, the value $Y_i$ that the steward gives to $\mathsf{O}$ is simply the query response $f_i(X_i)$ that it receives.
	\end{defn}
	
	Notice that this definition is more general than the usual notion of a pseudorandom generator in two respects. First, the query point $X_i$ may depend on the previous responses $f_1(X_1), \dots, f_{i - 1}(X_{i - 1})$. Second, rather than making a statistical indistinguishability requirement, we merely impose the standard steward correctness requirement, i.e. with high probability, every $Y_i$ has low $\ell_{\infty}$ error.
	
	We begin with an elementary lemma. Let $H(\cdot)$ denote Shannon entropy. If a random variable $X$ takes values in a set of size $t$, then $H(X) \leq \log_2 t$. The following lemma is a partial converse: if $H(X)$ is much smaller than $\log_2 t$, then there is a set of size $t$ that $X$ is likely to land in.
	\begin{lem} \label{lem:shannon}
		Suppose $X$ is a discrete random variable. Let $W$ be the set consisting of the $t$ most likely values of $X$. Then
		\[
		\Pr[X \in W] \geq 1 - \frac{H(X)}{\log_2 t}.
		\]
	\end{lem}
	
	\begin{proof}
		Let $p$ be the probability mass function of $X$, so that $H(X) = \E[-\log_2(p(X))]$. By Markov's inequality,
		\[
		\Pr[-\log_2(p(X)) > \log_2 t] \leq \frac{H(X)}{\log_2 t}.
		\]
		Therefore, if we let $W' = \{x: p(x) \geq 1/t\}$,
		\[
		\Pr[X \in W'] \geq 1 - \frac{H(X)}{\log_2 t}.
		\]
		Finally, $W' \subseteq W$, because there can be at most $t$ values $x$ such that $p(x) \geq 1/t$.
	\end{proof}
	
	\begin{thm}
		Suppose $\mathsf{S}$ is an $m$-coin pseudorandom generation $(\epsilon', \delta')$-steward for $k$ adaptively chosen $(\epsilon, \delta)$-concentrated functions $f_1, \dots, f_k: \{0, 1\}^n \to \R^d$. Then
		\[
		m \geq (1 - \delta')\cdot (n - \log_2(2/\delta)) \cdot k.
		\]
	\end{thm}
	
	\begin{proof}
		Without loss of generality, assume $d = 1$. Fix an injective $(\epsilon, 0)$-concentrated function $f: \{0, 1\}^n \to \R$. (E.g., we could set the binary expansion of $f(x)$ to be $\log(1/\epsilon)$ zeroes concatenated with $x$.) Let $\mathsf{O}_0$ be the owner who always chooses $f_i = f$ for every $i$. Let $X_i \in \{0, 1\}^n$ be the query point that $\mathsf{S}$ chooses in round $i$ of $\mathsf{O}_0 \leftrightarrow \mathsf{S}$, so that $X_1, \dots, X_k$ are functions of the internal randomness of $\mathsf{S}$.
		
		Since $\mathsf{S}$ uses only $m$ coins, $H(X_1, \dots, X_k) \leq m$. By the chain rule, this implies that there is some $j \in [k]$ such that
		\[
		H(X_j \mid X_1, \dots, X_{j - 1}) \leq \frac{m}{k}.
		\]
		Let $\mathsf{O}$ be the following owner:
		\begin{enumerate}
			\item In each round $i < j$, pick $f_i = f$, and obtain the value $Y_i$. Compute $x_i = f_i^{-1}(Y_i)$.
			\item Let $p(x)$ be the distribution $p(x) = \Pr[X_j = x \mid X_1 = x_1, \dots, X_{j - 1} = x_{j - 1}]$.
			\item In round $j$, let $W$ be the set of the $\lfloor \delta 2^n \rfloor$ points in $\{0, 1\}^n$ to which $p$ assigns the most mass. Let $f_j: \{0, 1\}^n \to \R$ be the function
			\[
			f_j(x) = \begin{cases}
			0 & \text{if } x \not \in W \\
			2\epsilon' & \text{if } x \in W.
			\end{cases}
			\]
			Observe that $f_j$ is $(0, \delta)$-concentrated at zero.
			Give $f_j$ to $\mathsf{S}$.
			\item In rounds $i > j$, choose $f_i$ to be (say) the constant zero function.
		\end{enumerate}
		
		The probability that $\mathsf{S}$ fails in $\mathsf{O} \leftrightarrow \mathsf{S}$ is precisely $\Pr[X_j \in W]$, which we now lower bound. We can write
		\[
		\Pr[X_j \in W] = \E_{(x_1, \dots, x_{j - 1}) \sim (X_1, \dots, X_{j - 1})}[\Pr[X_j \in W \mid X_1 = x_1, \dots, X_{j - 1} = x_{j - 1}]].
		\]
		By \cref{lem:shannon} and the definition of $W$, for any particular values $x_1, \dots, x_{j - 1}$,
		\[
		\Pr[X_j \in W \mid X_1 = x_1, \dots, X_{j - 1} = x_{j - 1}] \geq 1 - \frac{H(X_j \mid X_1 = x_1, \dots, X_{j - 1} = x_{j - 1})}{\log_2 \lfloor \delta 2^n \rfloor}.
		\]
		Therefore, by linearity of expectation,
		\begin{align*}
		\Pr[X_j \in W] &\geq 1 - \frac{H(X_j \mid X_1, \dots, X_{j - 1})}{\log_2 \lfloor \delta 2^n \rfloor} \\
		&\geq 1 - \frac{m}{k \log_2 \lfloor \delta 2^n \rfloor}.
		\end{align*}
		If $\delta < 2^{-n}$, the theorem statement is trivial, so assume that $\delta \geq 2^{-n}$. In this case, $\log_2 \lfloor \delta 2^n \rfloor \geq n - \log_2(2/\delta)$, so
		\[
		\Pr[X_j \in W] \geq 1 - \frac{m}{(n - \log_2(2/\delta)) \cdot k}.
		\]
		By the correctness of the steward, this value must be at most $\delta'$. Rearranging completes the proof.
	\end{proof}
	
	\subsection{Lower bound for one-query stewards}
	To understand the following lemma, imagine the perspective of $\mathsf{O}$ after $i - 1$ rounds of $\mathsf{O} \leftrightarrow \mathsf{S}(Z)$, where $Z$ was chosen uniformly at random from $\{0, 1\}^m$. Let $R$ be the set of $z$ such that the hypothesis that $Z = z$ is compatible with everything that $\mathsf{O}$ has seen so far. Then at this point, $\mathsf{O}$'s posterior distribution for $Z$ is uniform over $R$. The following lemma says that with respect to this posterior distribution, $\mathsf{O}$ can choose $f_i$ such that either $\mathsf{O}$ will learn $\Omega(1)$ bits of information about $Z$ based on $Y_i$, or else $\mathsf{S}$ will have a failure probability of $\Omega(1)$ in round $i$.
	
	\begin{lem} \label{lem:lower-bound}
		Suppose $\mathsf{S}$ is a one-query $m$-coin $(\epsilon', \delta')$-steward for $k$ adaptively chosen $(\epsilon, \delta)$-concentrated functions $f_1, \dots, f_k: \{0, 1\}^n \to \R$ and $\mathsf{O}$ is a deterministic owner. Fix $i \in [k]$. For a function $g: \{0, 1\}^n \to \R$, let $\mathsf{O}[g]$ be the owner that simulates $\mathsf{O}$ for rounds $1, 2, \dots, i - 1$, but chooses $g$ in round $i$ regardless of what $\mathsf{O}$ would have chosen. Let $R \subseteq \{0, 1\}^m$ be a nonempty set such that the transcript of the first $i - 1$ rounds of $\mathsf{O} \leftrightarrow \mathsf{S}(Z)$ is the same for every $Z \in R$. Assume $\delta \geq 2^{-n}$. Then there exists $g$ that is $(\epsilon, \delta)$-concentrated at $\mu$ such that either
		\begin{enumerate}
			\item $\max_{y \in \R} \Pr_{Z \in R}[Y_i = y \text{ in } \mathsf{O}[g] \leftrightarrow \mathsf{S}(Z)] \leq 0.8$, or \label{cond:learn}
			\item $\Pr_{Z \in R}[|Y_i - \mu_i| > \epsilon' \text{ in } \mathsf{O}[g] \leftrightarrow \mathsf{S}(Z)] \geq 0.2$. \label{cond:fail}
		\end{enumerate}
	\end{lem}
	
	\begin{proof}
		For each $j \in \Z$, let $g_j: \{0, 1\}^n \to \R$ be constant at $\epsilon j$. If some $g_j$ satisfies \cref{cond:learn}, we're done. So assume that for each $g_j$, there is some $y_j \in \R$ such that $\Pr_{Z \in R}[Y_j = y_j \text{ in } \mathsf{O}[g] \leftrightarrow \mathsf{S}(Z)] > 0.8$. If $y_j$ does not depend on $j$, then since $0.2 < 0.8$, there is some $g_j$ that satisfies \cref{cond:fail}, so we are again done. Therefore, assume there is some $j$ such that $y_j \neq y_{j + 1}$.
		
		Define $q: R \to \{0, 1\}^n$ by letting $q(Z) = $ the value $X_i$ chosen by $\mathsf{S}$ in $\mathsf{O} \leftrightarrow S(Z)$. First, assume there is some $x^*$ such that $\Pr_{Z \in R}[q(Z) = x^*] \geq 0.4$. For $s \in \{\pm 1\}$, define $g^{s}: \{0, 1\}^n \to \R$ by
		\[
			g^{s}(x) = \begin{cases}
				0 & \text{if } x = x^* \\
				s \cdot 2\epsilon' & \text{otherwise.}
			\end{cases}
		\]
		Then $g^s(x)$ is $(0, 2^{-n})$-concentrated at $s \cdot 2\epsilon'$. Let $\mathsf{O}'$ be the randomized owner that tosses a coin to decide whether to simulate $\mathsf{O}[g^{+1}]$ or $\mathsf{O}[g^{-1}]$. Then when $Z \in R$ is chosen uniformly at random, in $\mathsf{O}' \leftrightarrow \mathsf{S}(Z)$, there is a $0.4$ chance that $f_i(X_i) = 0$, in which case $\mathsf{S}(Z)$ has only a 50\% chance of correctly guessing $s$. This shows that $\Pr_{Z \in R}[|Y_i - \mu_i| > \epsilon' \text{ in } \mathsf{O}' \leftrightarrow \mathsf{S}(Z)] \geq 0.2$, and hence either $g^{+1}$ or $g^{-1}$ satisfies \cref{cond:fail}, so we are again done. Therefore, assume that for every $x^*$, $\Pr_{Z \in R}[q(Z) = x^*] < 0.4$.
		
		For $t \in \{j, j + 1\}$, let
		\[
			A_t = \{Z \in R: Y_i = y_t \text{ in } \mathsf{O}[g_t] \leftrightarrow \mathsf{S}(Z)\},
		\]
		so that $|A_t| > 0.8|R|$. We define $g$ by the following greedy algorithm. Two players, which we identify with $A_j$ and $A_{j + 1}$, alternate taking turns. When it is $A_t$'s turn, she finds the string $x \in \{0, 1\}^n$ such that $g(x)$ is not yet defined that maximizes $q^{-1}(x) \cap A_t$, and defines $g(x) = \epsilon t$. This continues for $2^n$ turns until $g$ is defined everywhere.
		
		Clearly, $g$ thus defined is $(\epsilon, 0)$-concentrated. We will show that $g$ satisfies \cref{cond:learn}. Proof: Say $z \in \{0, 1\}^m$ is \emph{good for $A_t$} if $z \in A_t$ and $g(q(z)) = \epsilon t$. In these terms, on $A_t$'s turn, she defines $g$ on one more point in order to maximize the number of $z$ that become good for $A_t$. Say that $z \in \{0, 1\}^m$ is \emph{bad for $A_t$} if $z \in A_t$ and $g(q(z)) \neq \epsilon t$. When it is not $A_t$'s turn, some $z$ may become bad for $A_t$, but the crucial point is that the number of $z$ that become bad for $A_t$ is \emph{at most} the number of $z$ that became \emph{good} for $A_t$ in the \emph{previous} turn (simply because of the greedy choice that $A_t$ made in the previous turn.) This would show that half of $A_t$ is good for $A_t$, except for one annoyance: the first turn, where some $z$ become bad for the second player with no corresponding previous turn, when $z$ became good. But we already showed that for every $x$, $|q^{-1}(x)| \leq 0.4 |R|$, so the first turn does not matter too much: at the end of the construction, for each $t$, the number of $z$ that are good for $A_t$ is at least
		\[
			\frac{1}{2}(|A_t| - 0.4|R|) \geq \frac{1}{2}(0.8 |R| - 0.4 |R|) = 0.2 |R|.
		\]
		By construction, if $z$ is good for $A_t$, then $Y_i = y_t$ in $\mathsf{O}[g] \leftrightarrow \mathsf{S}(z)$. Therefore, for each $t$, $\Pr_{Z \in R}[Y_i = y_t \text{ in } \mathsf{O}[g] \leftrightarrow \mathsf{S}(Z)] \geq 0.2$, which implies \cref{cond:learn} since $y_j \neq y_{j + 1}$.
	\end{proof}
	
	Having proved \cref{lem:lower-bound}, we are ready to prove our randomness complexity lower bound. The idea is that $\mathsf{O}$ will spend the first $k - 1$ rounds learning as much information as possible about $\mathsf{S}$'s randomness string using \cref{lem:lower-bound} (unless she gets lucky and is able to cause $\mathsf{S}$ to have an $\Omega(1)$ failure probability in one of these rounds, in which case she will take the opportunity.) Then, in round $k$, $\mathsf{O}$ uses everything she's learned about $\mathsf{S}$'s randomness string to choose $f_k$ so as to maximize $\mathsf{S}$'s failure probability in that round.
	
	\begin{thm} \label{thm:lower-bound}
		Suppose $\mathsf{S}$ is a one-query $m$-coin $(\epsilon', \delta')$-steward for $k$ adaptively chosen $(\epsilon, \delta)$-concentrated functions $f_1, \dots, f_k: \{0, 1\}^n \to \R^d$. Assume $\delta' < 0.2$ and $\delta \geq 2^{-n}$. Then $m \geq n + \Omega(k) - \log_2(\delta'/\delta)$.
	\end{thm}
	
	\begin{proof}
		Without loss of generality, assume $d = 1$. Let $\mathsf{O}$ be the following owner:
		\begin{enumerate}
			\item For $i = 1$ to $k$:
			\begin{enumerate}
				\item Let $y_1, y_2, \dots, y_{i - 1}$ be the responses received so far.
				\item Let $R \subseteq \{0, 1\}^m$ be the set of $z$ such that in $\mathsf{O} \leftrightarrow \mathsf{S}(z)$, $Y_j = y_j$ for every $j < i$. (By induction, we have already defined the behavior of $\mathsf{O}$ in rounds $1, 2, \dots, i - 1$, so $R$ is well-defined. In other words, $R$ is the set of $z$ that are compatible with what $\mathsf{O}$ has seen so far.)
				\item If $i < k$:
				\begin{enumerate}
					\item Choose $f_i = g$, where $g$ is the function guaranteed by \cref{lem:lower-bound}. (Again, $\mathsf{O}$ is already defined and deterministic for rounds $1, 2, \dots, i - 1$, so we can sensibly apply the lemma.)
				\end{enumerate}
				\item Otherwise, if $i = k$:
				\begin{enumerate}
					\item Pick $S \subseteq R$, $|S| = \min\{\lfloor \delta 2^n \rfloor, |R|\}$ uniformly at random, pick $s \in \{\pm 1\}$ independently and uniformly at random, and choose
					\[
						f_k(x) = \begin{cases}
							0 & \text{if } x \in q(S) \\
							s \cdot 2\epsilon' & \text{otherwise.}
						\end{cases}
					\]
					(Note that $f_k$ is $(0, \delta)$-concentrated at $s \cdot 2 \epsilon'$, because $|q(S)| \leq |S| \leq \delta 2^n$.)
				\end{enumerate}
			\end{enumerate}
		\end{enumerate}
		To analyze $\mathsf{O}$, in $\mathsf{O} \leftrightarrow \mathsf{S}$, say that $\mathsf{O}$ \emph{tries to win} in round $i$ if either $i = k$ or else $i < k$ and the function $f_i$ chosen satisfies \cref{cond:fail} in \cref{lem:lower-bound}. For a string $z \in \{0, 1\}^m$, let $w(z) \in [k]$ be the index of the first round in which $\mathsf{O}$ tries to win in $\mathsf{O} \leftrightarrow \mathsf{S}(z)$, and let $\tau(z)$ be the transcript of rounds $1, 2, \dots, w(z) - 1$ in $\mathsf{O} \leftrightarrow \mathsf{S}(z)$. Note that since $\mathsf{O}$ is deterministic in rounds $1, 2, \dots, k - 1$, $w(z)$ and $\tau(z)$ are not random variables. Define an equivalence relation on $\{0, 1\}^m$ by saying that $z \sim z'$ if and only if $\tau(z) = \tau(z')$. Say $\mathsf{O}$ uses $v$ random bits.  We first show that for each equivalence class $\bar{z}$,
		\begin{equation} \label{eqn:fail-prob}
			\Pr_{Z \in \bar{z}, V \in \{0, 1\}^v}[\|Y_{w(\bar{z})} - \mu_{w(\bar{z})}\|_{\infty} > \epsilon' \text{ in } \mathsf{O}(V) \leftrightarrow \mathsf{S}(Z)] \geq \min\{0.2, \delta \cdot (1/0.8)^{k - 1} \cdot 2^{n - m - 2}\}.
		\end{equation}
		Proof: Observe that in round $w(\bar{z})$, $\mathsf{O}$'s set $R$ is precisely $\bar{z}$. If $w(\bar{z}) < k$, then \cref{cond:fail} of \cref{lem:lower-bound} immediately implies that the failure probability in \cref{eqn:fail-prob} is at least $0.2$. Suppose instead that $w(\bar{z}) = k$. Then in every previous round, $\mathsf{O}$ did not try to win, i.e. $\mathsf{O}$ chose a function satisfying \cref{cond:learn} of \cref{lem:lower-bound}. This implies that in every previous round, $\mathsf{O}$'s set $R$ decreased in size by a factor of $0.8$. So at the beginning of round $k$, $|R| \leq 0.8^{k - 1} \cdot 2^m$. The probability (over $Z \in \bar{z}$) that $\mathsf{S}$ chooses $X_k$ such that $f_k(X_k) = 0$ is
		\begin{align*}
			\frac{|S|}{|R|} &= \frac{\min\{\lceil \delta 2^n \rceil, |R|\}}{|R|} \\
			&\geq \min\left\{1, \delta 2^{n - m - 1} (1/0.8)^{k - 1}\right\}.
		\end{align*}
		Conditioned on $f_k(X_k) = 0$, the probability of the event in \cref{eqn:fail-prob} is at least $0.5$, because conditioned on $f_k(X_k) = 0$, $s$ is independent of everything $\mathsf{S}$ has seen. Therefore, the probability of the event in \cref{eqn:fail-prob} is at least $\min\{0.5, \delta 2^{n - m - 2} (1/0.8)^{k - 1}\}$, completing the proof of \cref{eqn:fail-prob}.
		
		Now, to prove the theorem, observe that
		\begin{align*}
			\delta' &\geq \Pr_{Z \in \{0, 1\}^m, V \in \{0, 1\}^v}[\max_i \|\mu_i - Y_i\|_{\infty} > \epsilon' \text{ in } \mathsf{O}(V) \leftrightarrow \mathsf{S}(Z)] \\
			&\geq \Pr_{Z \in \{0, 1\}^m, V \in \{0, 1\}^v}[\|\mu_{w(Z)} - Y_{w(Z)}\|_{\infty} > \epsilon' \text{ in } \mathsf{O}(V) \leftrightarrow \mathsf{S}(Z)] \\
			&= \sum_{\bar{z}} \Pr_{Z \in \{0, 1\}^m} [Z \in \bar{z}] \cdot \Pr_{Z' \in \bar{z}, V \in \{0, 1\}^v} [\|\mu_{w(\bar{z})} - Y_{w(\bar{z})}\|_{\infty} > \epsilon' \text{ in } \mathsf{O}(V) \leftrightarrow \mathsf{S}(Z')] \\
			&\geq \sum_{\bar{z}} \Pr_{Z \in \{0, 1\}^m} [Z \in \bar{z}] \cdot \min\{0.2, \delta \cdot (1/0.8)^{k - 1} 2^{n - m - 2}\} \\
			&= \min\{0.2, \delta \cdot (1/0.8)^{k - 1} 2^{n - m - 2}\}.
		\end{align*}
		We assumed that $\delta' < 0.2$, so we can conclude that $\delta' \geq \delta \cdot (1/0.8)^{k - 1} 2^{n - m - 2}$. Rearranging proves that
		\begin{align*}
			m &\geq (n - 2) + (k - 1) \log_2(1/0.8) - \log_2(\delta'/\delta) \\
			&\geq n + \Omega(k) - \log_2(\delta'/\delta),
		\end{align*}
		completing the proof.
	\end{proof}
	
	\section{Directions for further research}
	
	The problem of randomness stewardship is fundamental, and the main open problem left by this work is to construct optimal stewards. The following are examples of concrete questions along these lines.
	\begin{itemize}
		\item Does every one-query steward with failure probability $\delta' \leq O(k \delta)$ have randomness complexity $n + \Omega(k \log(d + 1))$? (Is the randomness complexity of our main steward near-optimal?)
		\item Does there exist a one-query $(O(\epsilon), k\delta + 0.1)$-steward with randomness complexity $n + O(k \log(d + 1))$? (Can the error of our main steward be improved?)
	\end{itemize}
	We explained in this work how the steward model captures some older derandomization constructions, and we gave new applications of stewards. We hope that future researchers find more connections and applications.
	
	\section{Acknowledgments}
	
	We thank David Zuckerman for observations about block decision trees. We thank an anonymous reviewer for pointing out Impagliazzo and Zuckerman's previous work on this subject \cite{iz89, imp92}.
	
	\bibliographystyle{alpha}
	\bibliography{steward}
	
	\appendix
	
	\section{Generalized shifting and rounding algorithm} \label{apx:generalized-shifting-and-rounding}
	
	In this section, we show how to generalize the steward $\mathsf{S}_0$ to achieve a tradeoff between its error and the branching factor of the certification tree $T_{\mathsf{O}}$. Fix a factorization $d = d_0 d_1$. Partition $[d]$ as $[d] = J_1 \cup J_2 \cup \dots \cup J_{d_1}$, where $|J_t| = d_0$ for each $t$. Instead of partitioning $\R$ into intervals of length $2(d + 1) \epsilon$, partition $\R$ into intervals of length $2(d_0 + 1)\epsilon$. Let $\mathcal{I}$ denote the set of these intervals. The following algorithm for computing $Y_i$ from $W_i$ generalizes that of \cref{sec:base-construction}:
	\begin{enumerate}
		\item For each $t \in [d_1]$:
		\begin{enumerate}
			\item Find $\Delta_{it} \in [d_0 + 1]$ such that for every $j \in J_t$, there is a single interval in $\mathcal{I}$ that entirely contains $[W_{ij} + (2\Delta_{it} - 1)\epsilon, W_{ij} + (2\Delta_{it} + 1)\epsilon]$. (Such a $\Delta_{it}$ exists by \cref{lem:delta-exists}.)
			\item For every $j \in J_t$, set $Y_{ij} = \mathsf{Round}(W_{ij} + 2\Delta_{it} \epsilon)$.
		\end{enumerate}
	\end{enumerate}
	
	The following lemma is the appropriate generalization of \cref{lem:two-mediators}:
	\begin{lem} \label{lem:generalized-two-mediators}
		Assume $\delta < 1/2$. Let $\Sigma = [d_0 + 1]^{d_1} \cup \{\bot\}$. For any deterministic owner $\mathsf{O}$, there exists a $(k, n, \Sigma)$ block decision tree $T_{\mathsf{O}}$ with the following properties.
		\begin{enumerate}
			\item For any internal node $v$, $\Pr_{X \in \{0, 1\}^n}[v(X) = \bot] \leq \delta$. \label{cond:generalized-low-failure-prob}
			\item Fix $X_1, \dots, X_k \in \{0, 1\}^n$, and suppose that the path from the root to $T_{\mathsf{O}}(X_1, \dots, X_k)$ does not include any $\bot$ nodes. Then in $\mathsf{O} \leftrightarrow \mathsf{S}_0(X_1, \dots, X_k)$, $\max_i \|Y_i - \mu_i\|_{\infty} \leq O(d_0 \epsilon)$. \label{cond:generalized-model}
		\end{enumerate}
	\end{lem}
	
	The proof of \cref{lem:generalized-two-mediators} is essentially the same as the proof of \cref{lem:two-mediators}; we record the details below.
	
	\begin{proof}[Proof of \cref{lem:generalized-two-mediators}]
		For a vector $Y \in \R^d$ and function $f: \{0, 1\}^n \to \R^d$ that is $(\epsilon, \delta)$-concentrated at some point, say that a vector $(\Delta_1, \dots, \Delta_{d_1})$ is \emph{$f$-compatible} with $Y$ if $Y_j = \mathsf{Round}(\mu(f)_j + 2 \Delta_t \epsilon)$ for every $t \in [d_1]$ and every $j \in J_t$. Just for the analysis, let $\mathsf{S}_0'$ be the following (many-query) steward:
		\begin{enumerate}
			\item For $i = 1$ to $k$:
			\begin{enumerate}
				\item Give $f_i$ to $\mathsf{S}_0$, allowing it to make its one query and choose its output vector $Y_i \in \R^d$.
				\item Query $f_i$ at \emph{every} point in its domain, thereby learning the entire function.
				\item Compute
				\[
				\hat{\Delta}_i = \begin{cases}
				\text{the first } (\Delta_1, \dots, \Delta_{d_1}) \in [d_0 + 1]^{d_1} \text{ $f_i$-compatible with } Y_i & \text{if any exist} \\
				\bot & \text{otherwise.}
				\end{cases}
				\]
				\item Output $\hat{Y}_i = (\hat{Y}_{i1}, \dots, \hat{Y}_{id})$, where for each $t \in [d_1]$ and each $j \in J_t$,
				\[
				\hat{Y}_{ij} = \begin{cases}
				\mathsf{Round}(\mu(f_i)_j + 2\hat{\Delta}_{it} \epsilon) & \text{if } \hat{\Delta}_i \neq \bot \\
				0 & \text{otherwise.}
				\end{cases}
				\]
			\end{enumerate}
		\end{enumerate}
		
		The definition of $T_{\mathsf{O}}$ is exactly the same as in the proof of \cref{lem:two-mediators}, except that $\mathsf{S}_0'$ now refers to the above steward. To prove \cref{cond:generalized-low-failure-prob} in the lemma statement, we must show that in each round of $\mathsf{O} \leftrightarrow \mathsf{S}_0'$, $\Pr[\hat{\Delta}_i = \bot] \leq \delta$. Indeed, by concentration, with probability $1 - \delta$, for every $j$, $|W_{ij} - \mu(f_i)_{j}| \leq \epsilon$. In this case, by the construction of $\mathsf{S}_0$, $W_{ij} + 2\Delta_{it} \epsilon$ and $\mu(f_i)_j + 2\Delta_{it} \epsilon$ are in the same interval in $\mathcal{I}$ for every $t \in [d_1]$ and every $j \in J_t$. Therefore, in this case, there is at least one vector $(\Delta_1, \dots, \Delta_{d_1})$ that is $f_i$-compatible with $Y_i$, namely the vector of $\Delta_{it}$ values used by $\mathsf{S}_0$. To prove \cref{cond:generalized-model} in the lemma statement, suppose the path from the root node to $T_{\mathsf{O}}(X_1, \dots, X_k)$ does not include any $\bot$ nodes. Then in $\mathsf{O} \leftrightarrow \mathsf{S}_0'(X_1, \dots, X_k)$, for every $i$, $\hat{\Delta}_i \neq \bot$. This implies that every $Y_{ij}$ is of the form $\mathsf{Round}(\mu(f_i)_j + 2\hat{\Delta}_{it} \epsilon)$ for some $\hat{\Delta}_{it} \in [d_0 + 1]$. Therefore, $|Y_{ij} - \mu(f_i)_j| \leq 3(d_0 + 1) \epsilon$, since $2\hat{\Delta}_{it} \epsilon \leq 2(d_0 + 1) \epsilon$ and rounding introduces at most $(d_0 + 1) \epsilon$ additional error. Just as in the proof of \cref{lem:two-mediators}, the same bound holds in $\mathsf{O} \leftrightarrow \mathsf{S}_0$. Finally, since $\delta < 1/2$, $\|\mu(f_i) - \mu_i\|_{\infty} \leq 2\epsilon$, so by the triangle inequality, for every $i$, $\|Y_i - \mu_i\|_{\infty} \leq 3(d_0 + 1)\epsilon + 2\epsilon = (3d_0 + 5)\epsilon$.
	\end{proof}
	
	\section{The Saks-Zhou steward} \label{apx:sz}
	
	In this section, for completeness, we give the description and analysis of the Saks-Zhou steward. This algorithm and analysis are the same in spirit  as what appears in \cite{sz99}, but the presentation has been changed to match our framework. None of our results use this steward, but it is interesting to see how the stewards compare.
	
	\begin{prop} \label{prop:sz}
		For any $n, k, d \in \N$ and any $\epsilon, \delta, \gamma > 0$, there exists a one-query $(O(kd\epsilon/\gamma), k\delta + \gamma)$-steward for $k$ adaptively chosen $(\epsilon, \delta)$-concentrated functions $f_1, \dots, f_k: \{0, 1\}^n \to \R^d$ with randomness complexity
		\[
			n + O(k \log k + k \log d + k \log(1/\gamma)).
		\]
		The total running time of the steward is $\poly(n, k, d, \log(1/\epsilon), \log(1/\gamma))$.
	\end{prop}
	
	\begin{proof}
		Let $u$ be the smallest power of two such that $u \geq 2kd/\gamma$. (The only reason we choose a power of two is so that we can cleanly draw a uniform random element of $[u]$ using $\log u$ random bits.) Partition $\R$ into half-open intervals of length $\ell = u \epsilon$. For $w \in \R$, let $\mathsf{Round}(w)$ be the midpoint of the interval containing $w$. Algorithm $\mathsf{S}$:
		\begin{enumerate}
			\item Pick $X \in \{0, 1\}^n$ uniformly at random \emph{once}.
			\item For $i = 1$ to $k$:
			\begin{enumerate}
				\item Obtain $W_i = f_i(X) \in \R^d$.
				\item Pick $\Delta_i \in [u]$ uniformly at random.
				\item Return $Y_i = (Y_{i1}, \dots, Y_{id})$, where $Y_{ij} = \mathsf{Round}(W_{ij} + \Delta_i \epsilon)$.
			\end{enumerate}
		\end{enumerate}
		Proof of correctness: Just for the analysis, define a (many-query) steward $\mathsf{S}'$ by the following algorithm:
		\begin{enumerate}
			\item For $i = 1$ to $k$:
			\begin{enumerate}
				\item Query $f_i$ at \emph{every} point in its domain, thereby learning the entire function.
				\item Compute a point $\hat{\mu}_i \in \R^d$ where $f_i$ is $(\epsilon, \delta)$-concentrated.
				\item Pick $\Delta_i \in [u]$ uniformly at random.
				\item Return $Y_i = (Y_{i1}, \dots, Y_{id})$, where $Y_{ij} = \mathsf{Round}(\hat{\mu}_{ij} + \Delta_i \epsilon)$.
			\end{enumerate}
		\end{enumerate}
		Now fix any deterministic owner $\mathsf{O}$. For a vector $\vec{\Delta} = (\Delta_1, \dots, \Delta_k) \in [u]^k$, let $f_1^{[\vec{\Delta}]}, \dots, f_k^{[\vec{\Delta}]}$ be the functions that $\mathsf{O}$ chooses in $\mathsf{O} \leftrightarrow \mathsf{S}'(\vec{\Delta})$, and let $\hat{\mu}_i^{[\vec{\Delta}]}$ be the point at which $f_i^{[\vec{\Delta}]}$ is concentrated that $\mathsf{S}'$ chooses in $\mathsf{O} \leftrightarrow \mathsf{S}'(\vec{\Delta})$. Observe that
		\begin{equation} \label{eqn:sz-bad-event-1}
			\Pr_{\substack{\vec{\Delta} \in [u]^k \\ X \in \{0, 1\}^n}}[\text{for some $i$}, \|f_i^{[\vec{\Delta}]}(X) - \hat{\mu}_i^{[\vec{\Delta}]}\|_{\infty} > \epsilon] \leq k\delta.
		\end{equation}
		(Imagine picking $\vec{\Delta}$ first, and then apply the union bound over the $k$ different values of $i$.) Next, observe that
		\begin{equation} \label{eqn:sz-bad-event-2}
			\Pr_{\substack{\vec{\Delta} \in [u]^k \\ X \in \{0, 1\}^n}}[\text{for some $i, j$}, [\hat{\mu}_{ij}^{[\vec{\Delta}]} + (\Delta_i - 1) \epsilon, \hat{\mu}_{ij}^{[\vec{\Delta}]} + (\Delta_i + 1) \epsilon] \text{ is not entirely contained in one interval}] \leq \gamma.
		\end{equation}
		(Indeed, for each $i, j$, the probability is just $2/u$, so by the union bound, the probability is at most $2kd/u \leq \gamma$.) Now, by the union bound, assume from now on that $\vec{\Delta}, X$ are such that neither the event of \cref{eqn:sz-bad-event-1} nor the event of \cref{eqn:sz-bad-event-2} takes place. Assume without loss of generality that $\delta < 1/2$. We will show that in $\mathsf{O} \leftrightarrow \mathsf{S}(X, \vec{\Delta})$, for every $i$, $\|Y_i - \mu_i\|_{\infty} \leq 1.5 \ell + 3\epsilon$.
		
		We first show by induction on $i$ that in $\mathsf{O} \leftrightarrow \mathsf{S}(X, \vec{\Delta})$, every $f_i$ is precisely $f_i^{[\vec{\Delta}]}$. In the base case $i = 1$, this is trivial. For the inductive step, since the bad event of \cref{eqn:sz-bad-event-1} did not occur, we know that $f_i(X)$ is $\hat{\mu}_i^{[\vec{\Delta}]} \pm \epsilon$. Therefore, since the bad event of \cref{eqn:sz-bad-event-2} did not occur, for every $j$, $\mathsf{Round}(f_{ij}(X) + \Delta_i \epsilon) = \mathsf{Round}(\hat{\mu}_{ij}^{[\vec{\Delta}]} + \Delta_i \epsilon)$. Therefore, the value $Y_i$ in $\mathsf{O} \leftrightarrow \mathsf{S}(X, \vec{\Delta})$ is the same as the value $Y_i$ in $\mathsf{O} \leftrightarrow \mathsf{S}'(\vec{\Delta})$, and hence $\mathsf{O}$ chooses the same $f_{i + 1}$ in both cases. This completes the induction.
		
		Again using the fact that the bad event of \cref{eqn:sz-bad-event-1} did not occur, this immediately implies that in $\mathsf{O} \leftrightarrow \mathsf{S}(X, \vec{\Delta})$, every $f_i(X)$ is within $\ell_{\infty}$ distance $\epsilon$ of a point where $f_i$ is $(\epsilon, \delta)$-concentrated. Since $\delta < 1/2$, this implies that every $f_i(X)$ is within $\ell_{\infty}$ distance $3\epsilon$ of $\mu_i$. Shifting by $\Delta_i \epsilon$ and rounding introduce at most $1.5 \ell$ additional error, showing that $\|Y_i - \mu_i\| \leq 1.5 \ell + 3\epsilon$ as claimed. To complete the proof of correctness, note that $1.5 \ell + 3\epsilon \leq O(kd\epsilon/\gamma)$.
		
		The randomness complexity of this steward is $n$ bits (for $X$) plus the randomness needed for $\vec{\Delta}$, for a total randomness complexity of
		\[
			n + k \log u \leq n + O(k \log k + k \log d + k \log(1/\gamma)).
		\]
		The steward clearly runs in $\poly(n, k, d, \log(1/\epsilon), \log(1/\gamma))$ time.
	\end{proof}

	\section{The Impagliazzo-Zuckerman steward} \label{apx:iz}
	
	\subsection{Description and analysis of the steward}
	As discussed in \cref{sec:iz}, this steward is formed by combining the Impagliazzo-Zuckerman generator \cite{iz89, imp92} with straightforward rounding. None of our results use this steward.
	
	\begin{prop} \label{prop:iz}
		For any $n, k, d \in \N$ and $\epsilon, \delta > 0$, there exists a one-query $(O(\epsilon), k \delta + k \cdot 2^{-n^{\Omega(1)}})$-steward for $k$ adaptively chosen $(\epsilon, \delta)$-concentrated functions $f_1, \dots, f_k: \{0, 1\}^n \to \R^d$ with randomness complexity $O(n^6 + kd)$. The total running time of the steward is $\poly(n, k, d, \log(1/\epsilon))$.
	\end{prop}

	\begin{proof}
		Let $A$ be the ``bit provider'' of \cite[Theorem 2.6.8]{imp92}. Partition $\R$ into intervals of length $2\epsilon$. The steward:
		\begin{enumerate}
			\item For $i = 1$ to $k$:
			\begin{enumerate}
				\item Give $A$ the entropy bound $\delta n + d + 1$. It outputs some $X_i \in \{0, 1\}^n$.
				\item By querying, obtain the vector $W_i \stackrel{\text{def}}{=} f_i(X_i) \in \R^d$.
				\item Output $Y_i \in \R^d$, where $Y_{ij}$ is the midpoint of the interval containing $W_{ij}$.
			\end{enumerate}
		\end{enumerate}
		The randomness complexity of this steward is just the randomness complexity of $A$. From the proof of \cite[Theorem 2.6.8]{imp92}, this randomness complexity can be seen to be $O(n^6 + k/n + k(\delta n + d + 1))$. We may assume without loss of generality that $k \delta < 1$, because otherwise the proposition is trivial. Thus, the randomness complexity is $O(n^6 + kd)$ as claimed.
		
		Now, for correctness, suppose $f: \{0, 1\}^n \to \R^d$ is $(\epsilon, \delta)$-concentrated at $\mu \in \R^d$. Let $g(X)_j$ be the midpoint of the interval containing $f(X)_j$. Let $H(\cdot)$ denote Shannon entropy. To bound $H(g(U_n))$, let $S = \{x : \|f(x) - \mu\|_{\infty} \leq \epsilon\}$. Let $X \sim U_n$, and let $E$ indicate whether $X \in S$. Then by the chain rule,
		\begin{align*}
			H(g(X)) &= H(g(X) \mid E) + H(E) \\
			&\leq H(g(X) \mid X \in S) \cdot \Pr[X \in S] + H(g(X) \mid X \not \in S) \cdot \Pr[X \not \in S] + 1.
		\end{align*}
		Observe that $|g(S)| \leq 2^d$, because for each $j$, $[\mu_{ij} - \epsilon, \mu_{ij} + \epsilon]$ only intersects at most $2$ intervals. Therefore, $H(g(X) \mid X \in S) \leq d$. Furthermore, $H(X \mid X \not \in S) \leq n$ and applying a (deterministic) function can only reduce entropy, so $H(g(X) \mid X \not \in S) \leq n$. Therefore,
		\begin{align*}
			H(g(X)) &\leq d \cdot \Pr[X \in S] + n \cdot \delta + 1 \\
			&\leq \delta n + d + 1.
		\end{align*}
		By the guarantee of $A$, it follows that for any owner $\mathsf{O}$, the sequence of responses $(Y_1, Y_2, \dots, Y_k)$ is $(k \cdot 2^{-n^{\Omega(1)}})$-close (in total variation distance) from how it would be distributed if $A$ chose each $X_i$ independently and uniformly at random. In this ``fresh randomness'' case, by the union bound, with probability $1 - k \delta$, for every $i$, $\|W_i - \mu_i\|_{\infty} \leq \epsilon$. Rounding introduces at most $\epsilon$ additional error, so by the triangle inequality, in this case, for every $i$, $\|Y_i - \mu_i\|_{\infty} \leq 2\epsilon$.
	\end{proof}

	\subsection{Comparison to our stewards} \label{apx:iz-comparison}
	As noted in \cref{sec:iz}, our main steward (\cref{thm:main}) has much better randomness complexity than the Impagliazzo-Zuckerman steward. Furthermore, the Impagliazzo-Zuckerman steward has failure probability $\delta' = k \delta + k \cdot 2^{-n^{\Omega(1)}}$, which becomes trivial when $k$ is $O(2^{n^c})$ for a small constant $c > 0$. In contrast, our main steward has failure probability $\delta' = k\delta + \gamma$ for an arbitrary $\gamma > 0$; our steward is nontrivial for $\gamma$ as small as $2^{-\Omega(nk / \log k)}$ and $k$ as large as $\Omega(1/\delta)$, which could be as large as $2^{\Omega(n)}$.
	
	We also noted in \cref{sec:iz} that our main actually has worse accuracy than the Impagliazzo-Zuckerman steward. But one of our variant stewards, namely that of \cref{thm:prg-steward}, is strictly superior to the Impagliazzo-Zuckerman steward, as we now demonstrate:
	\begin{proof}[Proof of \cref{prop:iz} from \cref{thm:prg-steward}]
		Let $\mathsf{S}$ be the steward of \cref{thm:prg-steward} with $d_0 = 1$ and $\gamma = k \cdot 2^{-n^{1/2}}$. Then $\mathsf{S}$ is a one-query $(O(\epsilon), k\delta + k \cdot 2^{-n^{1/2}})$-steward. We may assume without loss of generality that $k \leq 2^{n^{1/2}}$, because otherwise \cref{prop:iz} is trivial (the failure probability bound is more than $1$). Therefore, the randomness complexity of $\mathsf{S}$ is bounded by
		\[
			n + O(kd + \log k \log(1/\gamma)) \leq O(n + kd). \qedhere
		\]
	\end{proof}
	
	\section{Nonconstructive PRG for block decision trees} \label{apx:prg}
	
	For completeness, we record the details of the standard nonconstructive argument that there exists a PRG for block decision trees with a small seed length.
	
	\begin{lem} \label{lem:nonconstructive-prg}
		Suppose $\mathcal{C}$ is a class of Boolean functions $f: \{0, 1\}^n \to \{0, 1\}$ such that a function in $\mathcal{C}$ can be specified using $t$ bits, i.e. $|\mathcal{C}| \leq 2^t$. Then for any $\gamma$, there exists a $\gamma$-PRG $\mathsf{Gen}: \{0, 1\}^s \to \{0, 1\}^n$ for $\mathcal{C}$ with seed length
		\[
			s \leq \log t + 2 \log(1/\gamma) + O(1).
		\]
	\end{lem}
	
	\begin{proof}
		Consider picking $\mathsf{Gen}$ uniformly at random from the set of all functions $\{0, 1\}^s \to \{0, 1\}^n$. Fix $C \in \mathcal{C}$, and let $\mu(C) = \Pr_x[C(x) = 1]$. Then for each fixed seed $x \in \{0, 1\}^s$, the probability (over $\mathsf{Gen}$) that $C(\mathsf{Gen}(x)) = 1$ is precisely $\mu(C)$. Therefore, the expected fraction of $x$ such that $C(\mathsf{Gen}(x)) = 1$ is precisely $\mu(C)$, and by Hoeffding's inequality,
		\[
			\Pr_{\mathsf{Gen}}\left[\left|\frac{\#\{x: C(\mathsf{Gen}(x)) = 1\}}{2^s} - \mu(C)\right| > \gamma\right] \leq 2^{-\Omega(\gamma^2 2^s)}.
		\]
		Therefore, by the union bound, the probability that the above bad event holds for \emph{any} $C$ is at most $2^{t - \Omega(\gamma^2 2^s)}$. If we choose $s$ large enough, this probability will be less than $1$, showing that there exists a $\mathsf{Gen}$ that works for all $C$. How large do we need to choose $s$? There is some constant $c$ such that it is sufficient to have $c \gamma^2 2^s > t$. Taking logarithms completes the proof.
	\end{proof}
	
	\begin{prop}
		For any $k, n \in \N$, any finite alphabet $\Sigma$, and any $\gamma > 0$, there exists a $\gamma$-PRG $\mathsf{Gen}: \{0, 1\}^s \to \{0, 1\}^{nk}$ for $(k, n, \Sigma)$f block decision trees with seed length
		\[
			s \leq n + k \log |\Sigma| + 2 \log(1/\gamma) + O(1).
		\]
	\end{prop}
	
	\begin{proof}
		Let $\mathcal{C}$ be the class of all Boolean functions $f: \{0, 1\}^{nk} \to \{0, 1\}$ of the form $f(x) = g(T(x))$, where $T$ is a $(k, n, \Sigma)$ block decision tree. To specify a function $f \in \mathcal{C}$, we need to specify (1) a bit for each leaf of $T$ and (2) a function $v: \{0, 1\}^n \to \Sigma$ for each internal node of $T$. In total, this number of bits $t$ is
		\begin{align*}
			t &= |\Sigma|^k + 2^n \lceil \log |\Sigma| \rceil \cdot \sum_{i = 0}^{k - 1} |\Sigma|^i \\
			&\leq |\Sigma|^k + 2^{n + 1} \log |\Sigma| \cdot \frac{|\Sigma|^k - 1}{|\Sigma| - 1} \\
			&\leq |\Sigma|^k + 2^{n + 1} |\Sigma|^k \\
			&\leq 2^{n + 2} |\Sigma|^k.
		\end{align*}
		By \cref{lem:nonconstructive-prg}, this implies that there is a $\gamma$-PRG $\mathsf{Gen}: \{0, 1\}^s \to \{0, 1\}^{nk}$ for $\mathcal{C}$ with seed length $n + k \log |\Sigma| + 2 \log(1/\gamma) + O(1)$. The ``operational'' characterization of total variation distance implies that $\mathsf{Gen}$ is also a $\gamma$-PRG for $(k, n, \Sigma)$ block decision trees as defined in \cref{sec:prg}.
	\end{proof}
\end{document}